\documentclass[aps,pra,twocolumn,superscriptaddress,longbibliography]{revtex4-2}

\usepackage{amssymb,amsmath,amsthm,bbm,bbold}
\usepackage{amsfonts}
\usepackage{graphicx}
\usepackage{mathtools}
\usepackage[breaklinks,colorlinks,allcolors=blue]{hyperref}
\usepackage[normalem]{ulem}
\usepackage[T1]{fontenc}
\usepackage{color}
\usepackage{xcolor}
\usepackage{amssymb}

\def\CC{\mathbbm{C}}

\newcommand{\Tr}{\mathrm{Tr}}

\newcommand{\bra}[1]{\mbox{$\langle #1 |$}}
\newcommand{\ket}[1]{\mbox{$| #1 \rangle$}}

\newtheorem{thm}{Theorem}

\newtheorem{cor}[thm]{Corollary}

\begin{document}
\title{Operational bounds and diagnostics for coherence in energy transfer}
\author{Julia Liebert}
\affiliation{Department of Chemistry, Princeton University, Princeton, NJ 08540, United States}
\author{Gregory D. Scholes}
\email{gscholes@princeton.edu}
\affiliation{Department of Chemistry, Princeton University, Princeton, NJ 08540, United States}

\begin{abstract}
Excitation energy transfer in light-harvesting aggregates is highly efficient, yet whether quantum coherence plays an operational role in transport remains debated. A central challenge is that coherence is usually inferred from spectroscopic signatures, whereas transport performance is assessed through specific observables and depends on both the open system dynamics and the initial state preparation.
Here we develop a resource theoretic approach that quantifies the maximum change that initial site-basis coherence can induce in a chosen readout under fixed reduced dynamics. The central quantity is the resource impact functional, which yields state independent, readout specific bounds on coherence-induced changes in signals and transport figures of merit.
We apply the framework to two models. For a donor-acceptor dimer, we analyse coherence sensitivity across coupling and bath-timescale regimes and bound trapping efficiency and average transfer time in terms of the impact functional. For a multi-site chain with terminal trapping, we derive rigorous criteria that distinguish population placement from sensitivity to initial state site-basis coherence. These include upper bounds on the largest advantage over incoherent preparations, necessary delocalization requirements for achieving a prescribed improvement, and a simple pairwise sufficient condition that can be checked from local information. For quasi-local reduced dynamics, we further obtain a Lieb-Robinson-type bound that constrains when coherence prepared in a distant region can influence a localized readout at finite times.
Together, these results provide operational diagnostics and rigorous bounds for benchmarking coherence effects and for identifying regimes in which they are necessarily negligible or potentially relevant in excitonic transport models.
\end{abstract}
\date{\today}   
\maketitle

\section{Introduction\label{sec:intro}}

Excitation energy transfer in molecular aggregates underpins photosynthetic light harvesting and has long served as a canonical setting for studying open quantum dynamics in complex condensed-phase environments \cite{SFOG11, MOAGGS17, SFCAGBC17, JM18, FS24, JZMTMD26}. Early ultrafast measurements reported oscillatory signals that were initially interpreted as evidence for coherent excitonic dynamics \cite{Jonas03, BMTSF04, ECRAMCBF07}. Subsequent work clarified that both the
microscopic origin and the functional significance of such oscillations can be subtle. Vibronic mixing and ground state  vibrational coherences can generate or prolong quantum beats in spectroscopic observables that may resemble purely electronic coherence \cite{CKPM12, TPJ13, CS15}. The excitation protocol also matters. Phase controlled ultrafast pulses can prepare coherent superpositions that are not directly comparable to excitation under natural incoherent sunlight
\cite{MV10, BS12, HSB13, CS15, CBSS15}. At the same time, coherent multidimensional spectroscopy is
a sensitive probe of excitonic Hamiltonians and system-bath couplings in pigment-protein complexes, and it provides benchmarks for microscopic models \cite{Cho08, CF09}.

These subtleties motivate shifting the focus from the mere presence of coherence to its operational impact on specific observables. A large body of theoretical work has addressed related issues using concrete models, including environment-assisted quantum transport \cite{MRLAG08, PH08, CCDHP09, RMKLAG09, KAG12}, disorder-dependent transport \cite{WRLGS97, MMG99, KMAMBG17}, and vibronic models \cite{WM11, LPCSLP15, JRWS18}. Quantitative conclusions, however, often depend on model choices,
parameter regimes, and assumed state preparations. This makes it difficult to give task-specific, state-independent statements that rigorously bound the maximum change that initial site-basis coherence can induce in a chosen readout. 

Here we develop a complementary, process-level perspective based on quantum resource theory. Building on the resource theory of coherence \cite{BCP14, WY16, CG19}, but without assuming that the physical dynamics is ``free'' in the resource-theoretic sense, we employ the operational framework introduced in Ref.~\cite{LS26-RT}. Specifically, we use as a central tool the resource impact functional $\mathcal{C}_M(\Lambda_t)$, which quantifies the largest change that initial site-basis coherence can induce in the expectation value of a chosen readout observable $M$ under a fixed dynamical map $\Lambda_t$. By construction, $\mathcal{C}_M(\Lambda_t)$ is state-independent and tailored to a specific readout. 

Although the framework developed in Ref.~\cite{LS26-RT} uses resource theoretic ingredients, it does not formulate a standard free operation resource theory of the dynamics.
The physical map $\Lambda_t$ is fixed and need not be free in the resource theoretic sense.
The role of the resource theory is instead to define a controlled resource-free counterpart $\mathcal G(\rho)$ of each input state $\rho$ via the use of a resource-destroying map $\mathcal{G}$ \cite{LHL17,Gour17,CG19}, which allows us to quantify the effect of removing the chosen resource from the initial preparation.
In this sense, the approach is a resource theoretic fixed-process diagnostic rather than a free operation achievability analysis.
Furthermore, the comparison between $\rho$ and $\mathcal G(\rho)$ is different from tests of nonclassical temporal behaviour, such as Leggett-Garg inequalities and multi-time coherence classicality diagnostics~\cite{LG85, WCM10, SEGSPH19, METTPSH20, GDDRECSSH20}.
These tests concern the compatibility of temporal statistics with a classical description, whereas $\mathcal C_M(\Lambda_t)$ quantifies the readout effect of coherence already present in the initial preparation under the fixed physical dynamics.

The resource impact functional then addresses the question whether coherence in the initial preparation can change the chosen readout under the same dynamics.
The comparison is between an initially prepared state and the state obtained by removing its site-basis coherences while keeping the same site populations.
This comparison is directly tied to situations in which coherent ultrafast preparations or initially delocalized excitonic states are contrasted with incoherent or site-localized population baselines \cite{MV10, BS12, CMM14, LMKT14}.

We apply this framework to excitation energy transfer with trapping. First, a donor-acceptor dimer is studied and the time-resolved impact functional is evaluated across coherent, incoherent, and intermediate system-bath coupling regimes.
In particular, we benchmark regimes where a time-local Lindblad description is appropriate and regimes where bath memory and intermediate coupling necessitate nonperturbative methods, which we treat using hierarchical equations of motion (HEOM) \cite{TK89, IF09-JCP, Tanimura20, IS20}. We show that $\mathcal{C}_M(\Lambda_t)$ faithfully captures the qualitative coherence trends across these regimes in a state-independent manner and, as a time-resolved quantity, identifies time windows in which coherence can affect the chosen readout most strongly. We further connect $\mathcal{C}_M(\Lambda_t)$ to standard transport figures of merit, including trapping efficiency and average transfer time, by expressing them in terms of effective measurement operators and deriving corresponding coherence-sensitive upper bounds.

At a coarse-grained level, aspects of exciton migration in photosynthetic antennae can be idealized by a one-dimensional chain with a terminal trap. In this geometry, two mechanisms can enhance the trapped population: (i) concentrating population near the trap and (ii) exploiting site-basis delocalization and interference during the transport. To quantify when coherence can and cannot be operationally relevant for trapping, we derive two complementary types of results expressed in terms of the resource impact functional: no-go bounds that certify regimes in which any coherence-induced improvement is necessarily negligible, and sufficient conditions under which coherent preparations can provably outperform incoherent (site localized) baselines.
In addition, we relate support-restricted impact functionals to quasi-locality estimates for open system dynamics, obtaining a Lieb-Robinson-type ``coherence light cone'' that bounds how rapidly coherence initially supported on one domain can influence a spatially separated local readout \cite{LR72, Poulin10, BK12, SEK19, TYR24}.

The paper is organized as follows. In Sec.~\ref{sec:concepts} we recap the definition and key properties of $\mathcal{C}_M(\Lambda_t)$ and explain how it can be constructed from HEOM trajectories. Sec.~\ref{sec:DAM} applies these tools to the donor-acceptor model across coupling and bath timescale regimes and relates $\mathcal{C}_M(\Lambda_t)$ to trapping figures of merit. Finally, in Sec.~\ref{sec:Kette} we discuss the multi-site donor chain and provide mathematical bounds relating delocalization and trapping performance via $\mathcal{C}_M(\Lambda)$, and the coherence light-cone analysis based on quasi-locality.

\section{Key concepts\label{sec:concepts}}

This section fixes notation and recalls the definitions needed in the sequel.
The central quantity is the resource impact functional $\mathcal{C}_M(\Lambda)$, which quantifies the largest change in the expectation value of an observable $M$ that can be attributed to a chosen resource present in the initial state and processed by a dynamical map $\Lambda$. We focus on the resource theory of coherence, motivated by excitation-energy transfer.
A key practical point is that $\mathcal{C}_M(\Lambda)$ can be evaluated from HEOM data even when the dynamical map $\Lambda$ is not available in closed form (in contrast to, e.g., a Lindblad generator). This will be used for the donor-acceptor model in the intermediate coupling regime in Sec.~\ref{sec:DAM}.

\subsection{Conceptual overview\label{sec:concept-overview}}

\begin{figure}[tb]
\centering
\includegraphics[width=\linewidth]{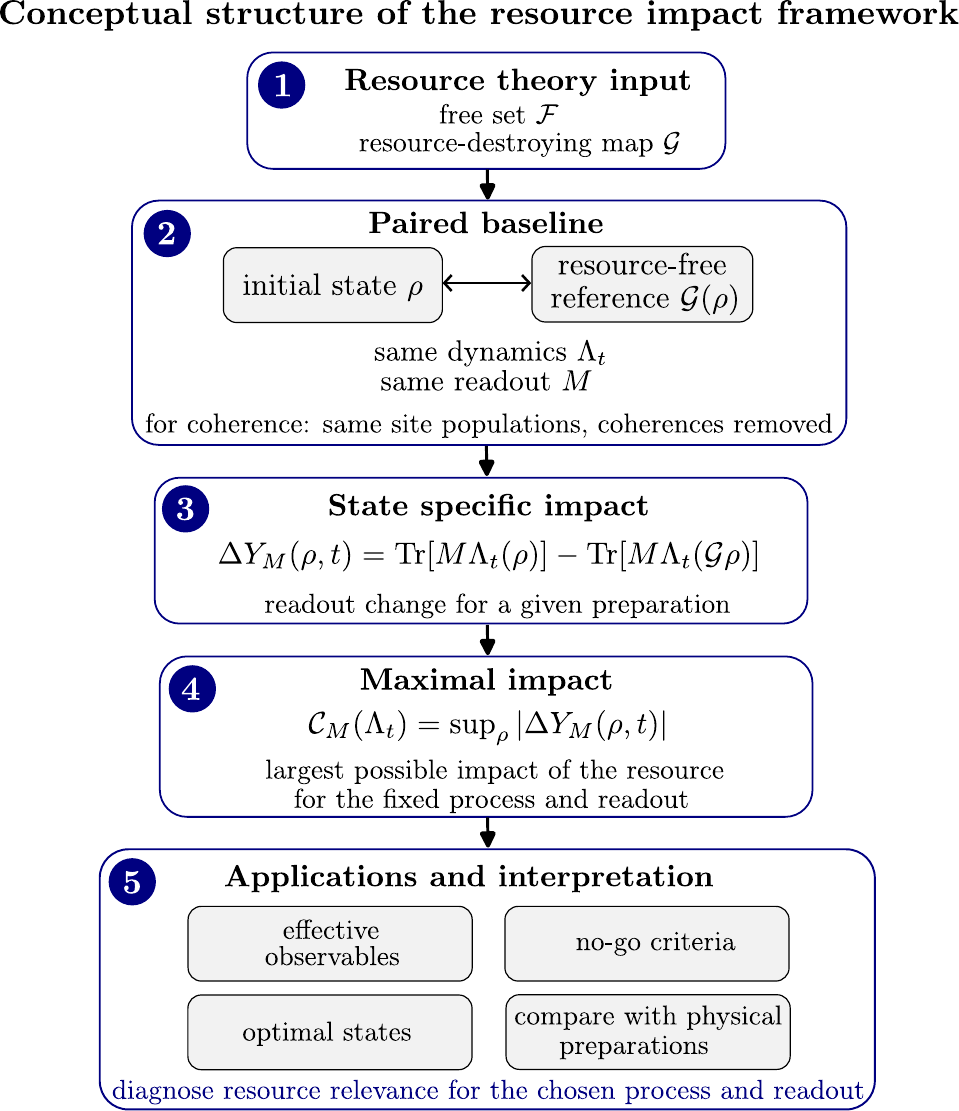}
\caption{Conceptual structure of the fixed-process resource-impact framework.
Each initial preparation $\rho$ is paired with the resource-free reference state $\mathcal G(\rho)$, and both states are evolved under the same physical dynamics $\Lambda_t$ and evaluated with the same readout $M$.
The resulting state specific readout difference $\Delta Y_M(\rho)$ measures the effect of removing the resource from that preparation.
Optimizing this difference over initial states gives the state-independent functional $\mathcal C_M(\Lambda_t)$ for the fixed process and readout.
The same structure leads to effective observables, no-go criteria, optimal resource-impact states, and comparisons with physically prepared states, which will be investigated in the context of excitation energy transfer in Secs.~\ref{sec:DAM} and \ref{sec:Kette}.}
\label{fig:RT-overview}
\end{figure}

To keep the paper self-contained, we briefly recall the resource impact framework of Ref.~\cite{LS26-RT} and specialize it to the coherence diagnostics used below.
Figure~\ref{fig:RT-overview} gives a schematic overview of the construction.
The starting point is a paired comparison.
For each initial state $\rho$, the resource-destroying map $\mathcal G$ defines a reference state $\mathcal G(\rho)$ in which the chosen resource has been removed.

There are two complementary levels of analysis.
The state-specific impact measures how much a chosen readout changes for a given preparation $\rho$ when the resource is removed.
The optimized quantity $\mathcal{C}_M(\Lambda)$ gives the largest possible readout change over all initial states for the fixed dynamics $\Lambda$ and readout $M$. We will recall the definition of $\mathcal{C}_M(\Lambda)$ in more detail in Sec.~\ref{sec:CM-definition}.
A crucial property of $\mathcal C_M(\Lambda)$ is that it is independent of the particular input state, but it remains specific to the physical process and observable under consideration. A second ingredient is the use of effective observables. 
When a transport readout can be written as a linear functional of the evolved state, one may moves the effect of the dynamics onto an observable acting on the initial state using properties of the resource impact functional shown in Ref.~\cite{LS26-RT}.
For instance, this will allow us to apply the resource impact functional to quantities such as transfer efficiency in Sec.~\ref{sec:efficiency-DAM}.
The same structure also yields no-go criteria.
If the effective observable is insensitive to the component removed by $\mathcal G$, then the resource cannot affect the readout under the fixed dynamics.
Finally, the framework identifies initial states with maximal resource impact for a given pair $(\Lambda,M)$.
Comparing these optimal states with physically motivated preparations, such as site localized states, delocalized excitonic states, or states prepared by optical excitation, shows how much of the possible resource impact is realized by those preparations.
This comparison helps distinguish transport performance that can already be explained by population level mechanisms from performance that is sensitive to initial coherence, and can guide the search for preparations for which the chosen process and readout make stronger use of the available resource.

\subsection{Resource impact functional\label{sec:CM-definition}}

Let $\mathcal{H}$ be a $d$-dimensional Hilbert space and let $\mathcal{D}(\mathcal{H})$ denote the set of density operators on $\mathcal{H}$. We consider a completely positive trace preserving (CPTP) map $\Lambda:\mathcal{D}(\mathcal{H})\to\mathcal{D}(\mathcal{H})$ and a resource-destroying map $\mathcal{G}$ associated with a given resource theory.
For a Hermitian observable $M=M^\dagger$, the resource impact functional introduced in Ref.~\cite{LS26-RT} is
\begin{equation}\label{eq:capacity-sup}
\mathcal{C}_M(\Lambda) := \sup_{\rho\in\mathcal{D}(\mathcal{H})} \Bigl| \Tr\Bigl[ M\,\Lambda \bigl(\rho-\mathcal{G}(\rho)\bigr)\Bigr] \Bigr| \geq 0\,,
\end{equation}
which measures the maximal change in $\langle M\rangle$ between evolving $\rho$ and evolving its resource-free counterpart $\mathcal{G}(\rho)$. 
Here, $\mathcal{G}$ is a resource-destroying map in the sense of Refs.~\cite{LHL17,Gour17}.
If $\mathcal{F}$ denotes the set of free states, then $\mathcal{G}(\rho)\in\mathcal{F}$ for every $\rho\in\mathcal{D}(\mathcal{H})$, and $\mathcal{G}(\omega)=\omega$ for every $\omega\in\mathcal{F}$.
Thus, $\mathcal{G}$ removes the chosen resource while leaving resource-free states unchanged.
Physically, $\mathcal G(\rho)$ is the paired reference preparation obtained by removing the chosen resource from $\rho$.
For site-basis coherence, which is the main focus of the remainder of this paper, the resource-destroying map is the dephased state with the same site populations, and will be defined in Eq.~\eqref{eq:dephasing-map}.
The pair $(\rho,\mathcal G(\rho))$ therefore compares preparations that differ only in their initial coherences, rather than unrelated coherent and incoherent states.
This provides a controlled baseline for assessing whether the resource present in the initial preparation changes the chosen readout.
For notational simplicity we suppress the explicit dependence of $\mathcal{C}_M(\Lambda)$ on $\mathcal{G}$.

The functional $\mathcal{C}_M(\Lambda)$ should be distinguished from a conventional resource monotone or from a bound on what can be achieved using only free operations.
The map $\Lambda$ is the physical dynamics under study and need not be free in the resource theoretic sense.
Resource theory enters here through the construction of the reference state $\mathcal{G}(\rho)$, rather than through a restriction of the dynamics.
In this sense, $\mathcal{C}_M(\Lambda)$ is a fixed-process and readout-specific diagnostic.

Eq.~\eqref{eq:capacity-sup} compares two initial preparations, $\rho$ and $\mathcal{G}(\rho)$, evolved under the same map $\Lambda$.
In the case of site-basis coherence, this keeps the initial site populations and the dynamics fixed while removing the initial coherences.
We use the terms resource-free, dephased, or incoherent baseline only in this state-level sense, namely for the preparation obtained by applying $\mathcal{G}$ to $\rho$.
Writing
\begin{equation}
\rho=\mathcal{G}(\rho)+\chi\,,\qquad \chi=(\mathrm{id}-\mathcal{G})(\rho)\,,
\end{equation}
the state-specific readout difference is
\begin{equation}
\Delta Y_M(\rho) = \Tr[M\Lambda(\rho)] -   \Tr[M\Lambda(\mathcal{G}(\rho))]   = \Tr[M\Lambda(\chi)]\,.
\end{equation}
Thus, $\mathcal{C}_M(\Lambda)$ maximizes the sensitivity of the chosen readout to the resource-dependent component $\chi$ present in the initial state.

For coherence, the reference evolution $\Tr[M\Lambda(\mathcal{G}(\rho))]$ already includes any coherences generated dynamically from the same initial populations.
The subtraction in $\Delta Y_M(\rho)$ therefore isolates the influence of coherences present in the initial preparation from effects that arise already in the evolution of the corresponding incoherent state.
In particular, if $\sigma=\mathcal{G}(\sigma)$, then $\Delta Y_M(\sigma)=0$, even if $\Lambda$ generates site-basis coherence from $\sigma$ and later converts it into populations.
Consequently, $\mathcal{C}_M(\Lambda)$ should not be interpreted as a measure of coherence-generating power~\cite{MK15,ZSCV17,ZSCV17-2}, or as a multi-time diagnostic of coherence-generating-and-detecting dynamics~\cite{SEGSPH19,GDDRECSSH20}.
This state level comparison is also distinct from tests of nonclassical temporal behaviour, which ask whether temporal statistics are compatible with a classical stochastic description~\cite{WCM10,SEGSPH19,METTPSH20}.
In the present framework, coherence plays a narrower operational role, where is relevant for the pair $(\Lambda,M)$ when removing it from the initial preparation changes the specified readout.

For linear CPTP resource-destroying maps as used in this paper, Eq.~\eqref{eq:capacity-sup} simplifies to 
\begin{equation}\label{eq:capacity}
\mathcal{C}_M(\Lambda) = \|B_{M,\Lambda}\|_{\infty}\,,\quad B_{M,\Lambda}:=(\mathrm{id}-\mathcal{G}^\dagger)\!\bigl(\Lambda^\dagger(M)\bigr)\,,
\end{equation}
where $\|\cdot\|_\infty$ denotes the operator (spectral) norm, which for Hermitian $B_{M,\Lambda}$ equals the eigenvalue of largest magnitude. Consequently, $\mathcal{C}_M(\Lambda)$ can be obtained by computing the extremal eigenvalues of $B_{M,\Lambda}$, without performing an optimization over states as in Eq.~\eqref{eq:capacity-sup}.

We now consider the resource theory of coherence \cite{BCP14, CG19}, and choose a reference basis $\mathcal{B}=\{\ket{i}\}_{i=1}^d$ (the site basis in our application), for which the free states are diagonal in $\mathcal{B}$. The corresponding resource-destroying map is the complete dephasing channel
\begin{equation}\label{eq:dephasing-map}
\mathcal{G}(\rho)=\sum_{i=1}^d \ket{i}\!\bra{i}\rho\ket{i}\!\bra{i}\,,
\end{equation} 
which removes all off-diagonal coherences in the chosen basis $\mathcal{B}$. This map is idempotent, $\mathcal{G}^2=\mathcal{G}$, as required for a resource-destroying map \cite{LHL17}, and it is self-adjoint, $\mathcal{G}^\dagger=\mathcal{G}$.

The following limiting cases help interpret $\mathcal C_M(\Lambda)$ before we turn to the explicit models in Secs.~\ref{sec:DAM} and \ref{sec:Kette}. They also help us to further understand the properties of the resource impact functional. 
For the coherence resource, where $\mathcal G$ is the complete dephasing map as defined in Eq.~\eqref{eq:dephasing-map}, the identity \eqref{eq:capacity}
shows that initial coherence can affect the readout only if the pulled-back observable $\Lambda^\dagger(M)$ has off-diagonal components in the reference basis.
In particular, for a diagonal readout $M$, any dynamics for which $\Lambda^\dagger(M)$ remains diagonal gives $\mathcal C_M(\Lambda)=0$.
This includes purely population-level dynamics in which output populations depend only on the initial populations, which means that $\mathcal G\circ\Lambda=\mathcal G\circ\Lambda\circ\mathcal G$, i.e., it is resource non-activating \cite{LHL17, LS26-RT}.
It also includes the ideal limit in which the dynamics effectively begins with complete dephasing, $\Lambda=\widetilde{\Lambda}\circ\mathcal G$, since then $\rho$ and $\mathcal G(\rho)$ lead to the same readout.
For large but finite dephasing rates, $\mathcal C_M(\Lambda)$ need not vanish exactly, since initial coherences may still produce short-time population currents before being damped, but the functional is expected to be small when such transients are suppressed.
In the opposite coherent limit, for unitary dynamics one has
$\Lambda_t^\dagger(M)=M+it[H,M]+O(t^2)$.
For diagonal $M$, initial coherence can then contribute at short times only if $(\mathrm{id}-\mathcal G)[H,M]\neq 0$.

For a one-parameter family of dynamical maps $\{\Lambda_t\}_{t\ge 0}$ with $\Lambda_0=\mathrm{id}$, we define the time-resolved resource impact functional $\mathcal C_M(t):=\mathcal C_M(\Lambda_t)$. 
The function $t\mapsto \mathcal C_M(t)$ quantifies, at each time $t$, the maximal change in the readout $\langle M\rangle_t$ that can be attributed to the chosen resource in the initial state, and thereby identifies time windows in which the observable $M$ is most sensitive to that resource.

For time-homogeneous Markovian dynamics described by a (time-independent) Gorini-Kossakowski-Sudarshan-Lindblad (GKSL) generator $\mathcal{L}$, the maps form a CPTP semigroup $\Lambda_t=e^{t\mathcal{L}}$ and $M_t:=\Lambda_t^\dagger(M)=e^{t\mathcal L^\dagger}(M)$ follows directly from the adjoint master equation $\tfrac{\mathrm{d}}{\mathrm{d}t} M_t=\mathcal L^\dagger(M_t)$ with $M_0=M$. In many energy-transfer settings, however, the Markovian assumptions underlying a Lindblad semigroup (weak coupling, short bath correlation time) are not justified, e.g., for strong system-bath coupling, structured environments, or pronounced memory effects. In such cases, one must resort to nonperturbative approaches. 
Here we employ the hierarchical equations of motion (HEOM) method, which yields the reduced dynamics through a linear time-local evolution in an extended space of auxiliary density operators (ADOs). While $\Lambda_t$ is typically not available in closed form, $\mathcal C_M(t)$ can be obtained by reconstructing the Heisenberg-evolved observable $M_t=\Lambda_t^\dagger(M)$ from a set of forward Schr\"odinger picture trajectories, as described next.

\subsection{Construction of $\mathcal{C}_M(\Lambda_t)$ for HEOM \label{sec:heom-CM}}

Due to Eq.~\eqref{eq:capacity}, it suffices to reconstruct $M_t:=\Lambda_t^\dagger(M)$ rather than the full superoperator $\Lambda_t$. For a time-independent system Hamiltonian and a stationary bath, HEOM yields a linear equation of motion for the reduced density operator together with a finite set of ADOs up to hierarchy depth $L_{\max}$. Collecting all ADOs in a vector $\mathbf R_t$ (with the time-dependent system's density operator $\rho_t$ as the first component) one obtains \cite{TK89, IF09-JCP, KKRH11, SS12, IS20, Tanimura20}
\begin{equation}
\frac{d}{dt}\mathbf R_t=\mathcal L_{\mathrm{HEOM}}\,\mathbf R_t\,, \quad
\mathbf R_t=e^{t\mathcal L_{\mathrm{HEOM}}}\mathbf R_0\,.
\end{equation}
For a factorized initial total state the initial ADOs vanish and $\mathbf R_0=\iota(\rho_0):=(\rho_0,0,\dots,0)$. Let $\Pi$ denote the projection onto the first component, so that $\rho_t=\Pi(\mathbf R_t)$. The induced reduced map is
\begin{equation}
\Lambda_t=\Pi\circ e^{t\mathcal L_{\mathrm{HEOM}}}\circ \iota\,.
\end{equation}
Although the adjoint map can be written explicitly in terms of $\mathcal L_{\mathrm{HEOM}}^\dagger$, in practice we reconstruct $M_t=\Lambda_t^\dagger(M)$ from forward trajectories via the duality relation
\begin{equation}\label{eq:Mt-duality}
\Tr[M_t X]=\Tr[\Lambda_t^\dagger(M)X]=\Tr[M \Lambda_t(X)]\,, \,\, X\in\mathcal B(\mathcal H)\,,
\end{equation}
where $\mathcal{H}$ denotes the system's Hilbert space without environment degrees of freedom. 
Choosing $X=\ket{k}\!\bra{l}$ yields the matrix elements
\begin{equation}\label{eq:Mt-elements}
(M_t)_{lk} = \Tr \big[M \Lambda_t(\ket{k}\!\bra{l})\big]\,.
\end{equation}
Assuming $M$ is Hermitian (as for physical observables), $M_t$ is Hermitian for Hermiticity-preserving dynamics, so it suffices to reconstruct, e.g., the upper-triangular elements.

A convenient informationally complete set of $d^2$ density operators is
\begin{equation}
\sigma_{kk}:=\ket{k}\!\bra{k}\,,\quad k=1,\dots,d\,,
\end{equation}
and for all $k<l$,
\begin{align}
\sigma_{kl}^{(+)} &:= \frac{(\ket{k}+\ket{l})(\bra{k}+\bra{l})}{2}\,,\\
\sigma_{kl}^{(i)} &:= \frac{(\ket{k}+i\ket{l})(\bra{k}-i\bra{l})}{2}\,.
\end{align}
For each initial state $\sigma$ we propagate $\rho_t=\Lambda_t(\sigma)$ with HEOM and record
\begin{align}
y_k(t) &:= \Tr \big[M\Lambda_t(\sigma_{kk})\big]\,,\\
y_{kl}^{(+)}(t) &:= \Tr \big[M\Lambda_t(\sigma_{kl}^{(+)})\big]\,,\,\,
y_{kl}^{(i)}(t) := \Tr \big[M\Lambda_t(\sigma_{kl}^{(i)})\big]\,.
\end{align}
Defining
\begin{align}
c_{kl}^{(1)}(t) &:= 2 y_{kl}^{(+)}(t) - y_k(t) - y_l(t)\,,\\
c_{kl}^{(2)}(t) &:= 2 y_{kl}^{(i)}(t) - y_k(t) - y_l(t)\,,
\end{align}
the matrix elements of $M_t$ in the basis $\{\ket{k}\}$ are
\begin{equation}
(M_t)_{kk}=y_k(t)\,,\quad (M_t)_{kl}=\frac{c_{kl}^{(1)}(t)- i\,c_{kl}^{(2)}(t)}{2}\quad (k<l)\,.
\end{equation}
This reconstructs $M_t$ (and hence $\mathcal C_M(t)$) without explicitly forming the full superoperator $\Lambda_t$.

\section{Donor-acceptor model\label{sec:DAM}}

We now apply the time-resolved resource impact functional $\mathcal{C}_M(\Lambda_t)$ to a minimal donor-acceptor dimer model of excitation energy transfer with recombination and irreversible trapping in a sink coupled to the acceptor site. Despite involving only two coherently coupled sites, this spin-boson type excitation energy transfer model serves as a common setting for exploring the interplay of coherent mixing, environmental fluctuations, and trapping \cite{IF09-Redfield,IF09-JCP,IF12}.

Each chromophore $i=1,2$ is modelled as a two-level system with ground and excited states $\ket{i_g}$ and $\ket{i_e}$. We restrict to the zero- and single-excitation manifold spanned by the overall ground state $\ket{g}:=\ket{1_g}\otimes\ket{2_g}$ and the single-excitation states $\ket{D}:=\ket{1_e}\otimes\ket{2_g}$ and $\ket{A}:=\ket{1_g}\otimes\ket{2_e}$ as illustrated in Fig.~\ref{fig:DAM-schematic}. To quantify trapping efficiency while keeping the reduced dynamics trace preserving, we include an additional absorbing sink state $\ket{s}$, so that $\mathcal{H}_S=\mathrm{span}\{\ket{g},\ket{D},\ket{A},\ket{s}\}$. Recombination from $\ket{D}$ and $\ket{A}$ to $\ket{g}$ occurs at rates $\Gamma_D$ and $\Gamma_A$,
respectively. Trapping into the sink occurs from $\ket{A}$ at rate $\kappa_A$ and is taken to be irreversible (no back-coupling from $\ket{s}$ to $\ket{A}$). We focus on acceptor-only trapping here, as an asymmetric trap is the relevant scenario for the environment-assisted quantum transport discussion in Sec.~\ref{sec:efficiency-DAM} in the spirit of Refs.~\cite{MRLAG08, PH08, CCDHP09, RMKLAG09, KAG12}.

\begin{figure}[tb]
\centering
\includegraphics[width=0.75\linewidth]{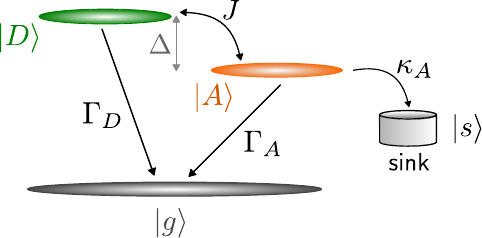}
\caption{Donor-acceptor dimer with recombination to the joint ground state and an absorbing sink
attached to the acceptor. The donor and acceptor are additionally coupled to site-local bosonic
environments (not shown). \label{fig:DAM-schematic}}
\end{figure}

Within the $\{\ket{D},\ket{A}\}$ manifold, the coherent dynamics is generated by
\begin{equation}\label{eq:Hex}
H_{\mathrm{ex}} =   \frac{\Delta}{2}\!\left(\ket{D}\!\bra{D}-\ket{A}\!\bra{A}\right) +J \left(\ket{D}\!\bra{A}+\ket{A}\!\bra{D}\right)\,,
\end{equation}
where $\Delta$ is the site-energy detuning (so that $E_D-E_A=\Delta$ in the absence of coupling) and
$J$ is the electronic coupling. Defining the mixing angle $\theta$ by $\tan(2\theta)=2J/\Delta$, the
eigenstates are
\begin{equation}
\ket{+} = \cos\theta\,\ket{D}+\sin\theta\,\ket{A}\,, \,\, \ket{-} = -\sin\theta\,\ket{D}+\cos\theta\,\ket{A}\,,
\end{equation}
with eigenenergies $\pm\frac{1}{2}\sqrt{\Delta^2+4J^2}$.

Moreover, the donor and acceptor excited states are coupled to a bosonic environment that modulates the site energies. These environmental fluctuations dephase coherences in the site basis and, because the coupling is generally not diagonal in the exciton eigenbasis, also drive population relaxation between exciton states through dissipative bath dynamics. Throughout the remainder of this section, we use a Drude-Lorentz (overdamped Brownian oscillator) spectral density,
\begin{equation}
J_{\mathrm{DL}}(\omega)=\frac{2E_R\,\gamma_c\,\omega}{\gamma_c^2+\omega^2},
\end{equation}
where $E_R$ sets the overall coupling strength and $\gamma_c$ is the cutoff (bath relaxation) rate, corresponding to a correlation time $\tau_c=1/\gamma_c$.

\subsection{Relevant coupling regimes and time scales \label{sec:DAM-parameter-regimes}}

To connect this minimal model to molecular excitation-energy transfer, we organize the parameter space in terms of the relative strength of coherent mixing and environmental reorganization, as well as the degree of bath memory. Following common practice in the excitation energy transfer literature, it is convenient to use two dimensionless ratios: (i) $J/E_R$, comparing electronic coupling $J$ to the reorganization energy $E_R$, and (ii) $\tau_c/J^{-1}$, comparing the bath correlation time $\tau_c$ to the intrinsic coherent timescale $J^{-1}$ \cite{IF09-JCP,IF12}. Here $E_R$ denotes the reorganization energy, i.e., the energy dissipated as the environment relaxes following a vertical (Franck-Condon) excitation. The bath correlation time $\tau_c$ characterizes the decay of environmental correlations and thus the relevance of memory effects. 
Intermediate regimes in which $E_R$ is not small compared to $J$ and/or $\tau_c$ is not much shorter than $J^{-1}$ are widely encountered in photosynthetic EET and motivate treatments beyond standard perturbative Markovian master equations \cite{IF12,CS15,LLHC16}.

In the limit $\tau_c\ll J^{-1}$ (rapidly decaying bath correlations) together with $J/E_R\gg 1$ (weak environmental perturbation relative to coherent mixing), Born-Markov and Redfield-type approaches are expected to be reliable. 
Conversely, when $J\sim E_R$ and/or $\tau_c\sim J^{-1}$, bath reorganization dynamics and non-Markovian effects become important and Markovian reductions can fail quantitatively \cite{IF09-JCP,IF12}. We therefore evaluate $\mathcal{C}_{M_A}(\Lambda_t)$ using HEOM for a system linearly coupled to bosonic environments as explained in Sec.~\ref{sec:heom-CM}. 

Since the optimization in $\mathcal{C}_{M_A}(\Lambda_t)$ is restricted to initial states supported on
the donor-acceptor subspace $\mathcal{H}_{DA}=\mathrm{span}\{\ket{D},\ket{A}\}$, only the $2\times 2$ block of $M_t$ on $\mathcal{H}_{DA}$ is required. 
In particular, for $M_A=\ket{A}\!\bra{A}$ and site dephasing $\mathcal G$ on $\{\ket{D},\ket{A}\}$, $(\mathrm{id}-\mathcal G^\dagger)(M_t)$ removes the diagonal components so that $\mathcal C_{M_A}(t)$ depends only on the off-diagonal element $(M_t)_{DA}$ (indeed, $\mathcal C_{M_A}(t)=|(M_t)_{DA}|$ in the two-level case). This element can be reconstructed from four trajectories corresponding to the initial pure states $\ket{D}$, $\ket{A}$, $(\ket{D}+\ket{A})/\sqrt2$, and $(\ket{D}+i\ket{A})/\sqrt2$, reducing the computational cost relative to a full $d^2$-element reconstruction. The numerics is performed using the QuTiP HEOM implementation \cite{LRCMAPBN23}.

\subsection{$C_M(\Lambda_t)$ for different parameter regimes \label{sec:DAM-int}}

\begin{figure*}[tb]
\centering
\includegraphics[width=\linewidth]{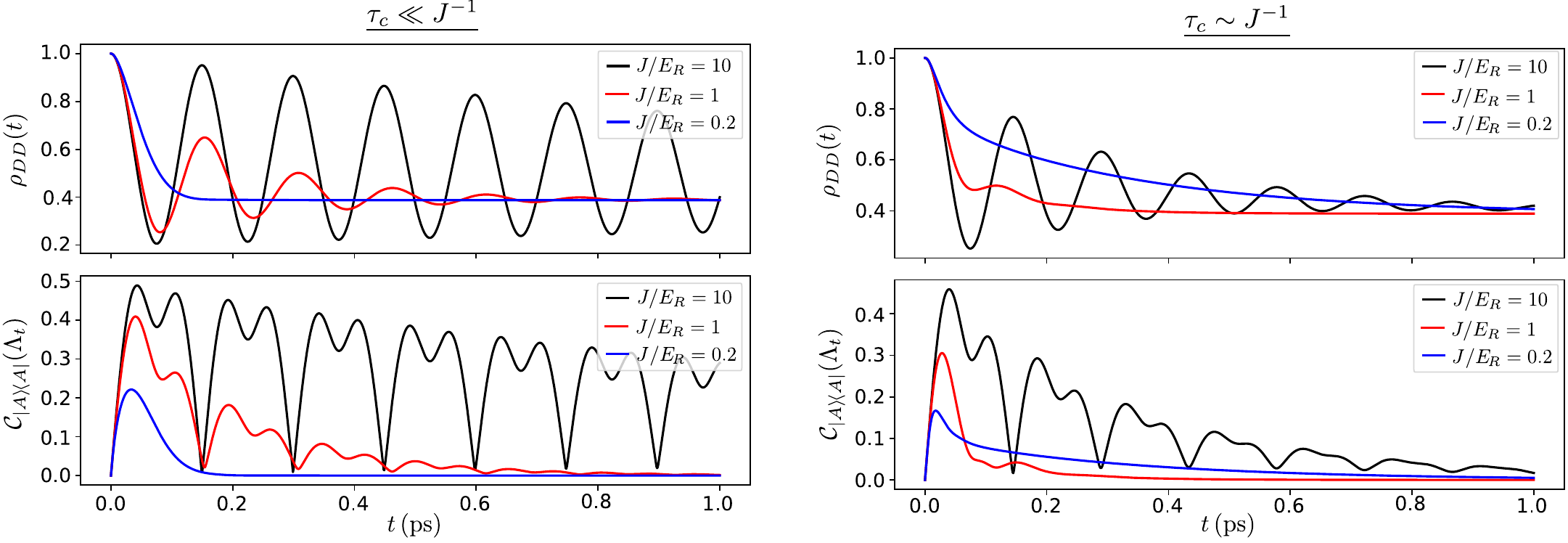}
\caption{Comparison between the donor population $\rho_{DD}(t)$ and $\mathcal{C}_{|A\rangle\!\langle A|}(\Lambda_t)$ for $J = 100\,\mathrm{cm}^{-1}$, $\Delta = 100\,\mathrm{cm}^{-1}$ and $T=300\,\mathrm{K}$, comparing a fast-bath setting $\tau_c \ll J^{-1}$ (left, $\tau_c = 2.65\,\mathrm{fs}$) with an intermediate-timescale setting $\tau_c \sim J^{-1}$ (right, $\tau_c = 53.1\,\mathrm{fs}$) for different ratios $J/E_R\in\{10,1,0.2\}$ at fixed $J$. 
\label{fig:DAM-CM-Markov}}
\end{figure*}

We first consider $M =\ket{A}\!\bra{A}$, such that $\Tr[M\Lambda_t(\rho_0)]$ is the acceptor population at time $t$ and $\rho_0$ is the initial state. In this language, $\mathcal{C}_M(\Lambda)$ quantifies the maximal increase in acceptor yield that any initial site coherence can provide, compared to its dephased counterpart $\mathcal{G}(\rho)$.

Figure~\ref{fig:DAM-CM-Markov} compares the donor population $\rho_{DD}(t)=\langle D|\rho(t)|D\rangle$ for an initially donor-localised excitation $\rho(0)=\ket{D}\!\bra{D}$ with the time-resolved impact functional $\mathcal{C}_{|A\rangle\!\langle A|}(\Lambda_t)$, both for a fast-bath choice $\tau_c\ll J^{-1}$ (left, $\tau_c=2.65\,\mathrm{fs}$) and for an intermediate-timescale choice $\tau_c\sim J^{-1}$ (right, $\tau_c=53.1\,\mathrm{fs}$). We fix $J=100\,\mathrm{cm}^{-1}$ and vary $E_R$ to access $J/E_R\in\{10,1,0.2\}$, using $\Delta=100\,\mathrm{cm}^{-1}$ and $T=300\,\mathrm{K}$. To isolate the role of coherent mixing and bath-induced decoherence without additional irreversible loss, we set $\Gamma_A=\Gamma_D=\kappa_A= 0$. 
The HEOM hierarchy is truncated at depth eight and we retain $N_k=6$ terms in the (Matsubara) expansion of the bath correlation function. 

Recall that $\mathcal{C}_{M_A}(\Lambda_t)$ quantifies the maximal change in acceptor population attributable to coherence in the initial state, relative to the dephased baseline $\mathcal{G}(\rho_0)$. In particular, if $\mathcal{C}_{|A\rangle\!\langle A|}(\Lambda_t)>0$, then the maximiser in Eq.~\eqref{eq:capacity-sup} cannot be an incoherent state such as $\rho_0=\ket{D}\!\bra{D}$, since $\rho_0=\mathcal{G}(\rho_0)$ implies that the yield difference $\Delta Y_t(\rho_0) : =\Tr[M\Lambda_t(\rho_0 - \mathcal{G}(\rho_0))] = 0$ entering Eq.~\eqref{eq:capacity-sup} vanishes. Due to Eq.~\eqref{eq:capacity} any oscillatory structure in $\mathcal{C}_{M_A}(\Lambda_t)$ is induced by the dynamical map $\Lambda_t$ itself. Conversely, when $\mathcal{C}_{M_A}(\Lambda_t)\approx 0$, the acceptor population at that time is approximately insensitive to initial site-basis coherence, consistent with the comparison to $\rho_{DD}(t)$ in Fig.~\ref{fig:DAM-CM-Markov}.

There are two messages:
First, over the parameter range considered, $\mathcal{C}_{M_A}(\Lambda_t)$ follows the qualitative behaviour reported in state-specific analyses of related dimer models \cite{IF09-Redfield, IF09-JCP, IF10-PCCP, IF12}: stronger system-bath couplings (smaller $J/E_R$) and/or towards longer bath memory (larger $\tau_c/J^{-1}$) lead to a faster decaying envelope of $\mathcal{C}_{M_A}(\Lambda_t)$. 
Since $\mathcal{C}_{M_A}(\Lambda_t)$ upper-bounds the maximal coherence induced change in the acceptor population, this faster decay can be interpreted as the dynamics becoming progressively less able to exploit initial site-basis coherence to modulate the acceptor population.
Second, because $\mathcal{C}_{M_A}(\Lambda_t)$ is time-resolved and tied to the readout $M_A$, it directly identifies the time windows in which initial coherence can, in principle, have the largest impact on the acceptor population.
This makes $\mathcal{C}_{M_A}(\Lambda_t)$ a convenient diagnostic for separating parameter regimes where coherence effects are potentially relevant from those where they are provably negligible for a given observable.

In experimentally relevant scenarios, non-free initial states may arise, for instance, from ultrafast excitation preparing coherent superpositions in the site basis or from initially delocalised excitons spanning donor and acceptor sites. For a given preparation $\rho_0$, the utilisation ratio $|\Delta Y_t(\rho_0)|/\mathcal{C}_M(\Lambda_t)\leq 1$ quantifies how closely the realised coherence effect approaches the corresponding upper bound. For restricted families of chemically reachable states, two natural adaptations are (i) restricting the optimisation in Eq.~\eqref{eq:capacity-sup} to $\mathcal{R}(\mathcal{H})\subseteq\mathcal{D}(\mathcal{H})$,
\begin{equation}
\bar{\mathcal{C}}_{M}^{(\mathcal{R})}(\Lambda):= \sup_{\rho\in \mathcal{R}(\mathcal{H})} |\Delta Y(\rho)| \leq \mathcal{C}_M(\Lambda)\,,
\end{equation}
or (ii) optimising over an admissible set $\mathcal P_{\mathrm{prep}}$ of preparation maps $P$ acting on a fixed state $\rho$,
\begin{equation}
\bar{\mathcal{C}}_{M}^{(\mathcal{P})}(\Lambda):= \sup_{P\in\mathcal P_{\mathrm{prep}}} \Big|\Tr \big[M\Lambda\!\big(P(\rho)-\mathcal G(P(\rho))\big)\big]\Big| \leq \mathcal{C}_M(\Lambda).
\end{equation}
In both cases the convenient operator norm expression available for linear $\mathcal G$ is generally lost, whereas the original $\mathcal{C}_M(\Lambda)$ remains straightforward to evaluate without an explicit optimisation over input states as in Eq.~\eqref{eq:capacity}.

\begin{figure*}[tb]
\centering
\includegraphics[width=\linewidth]{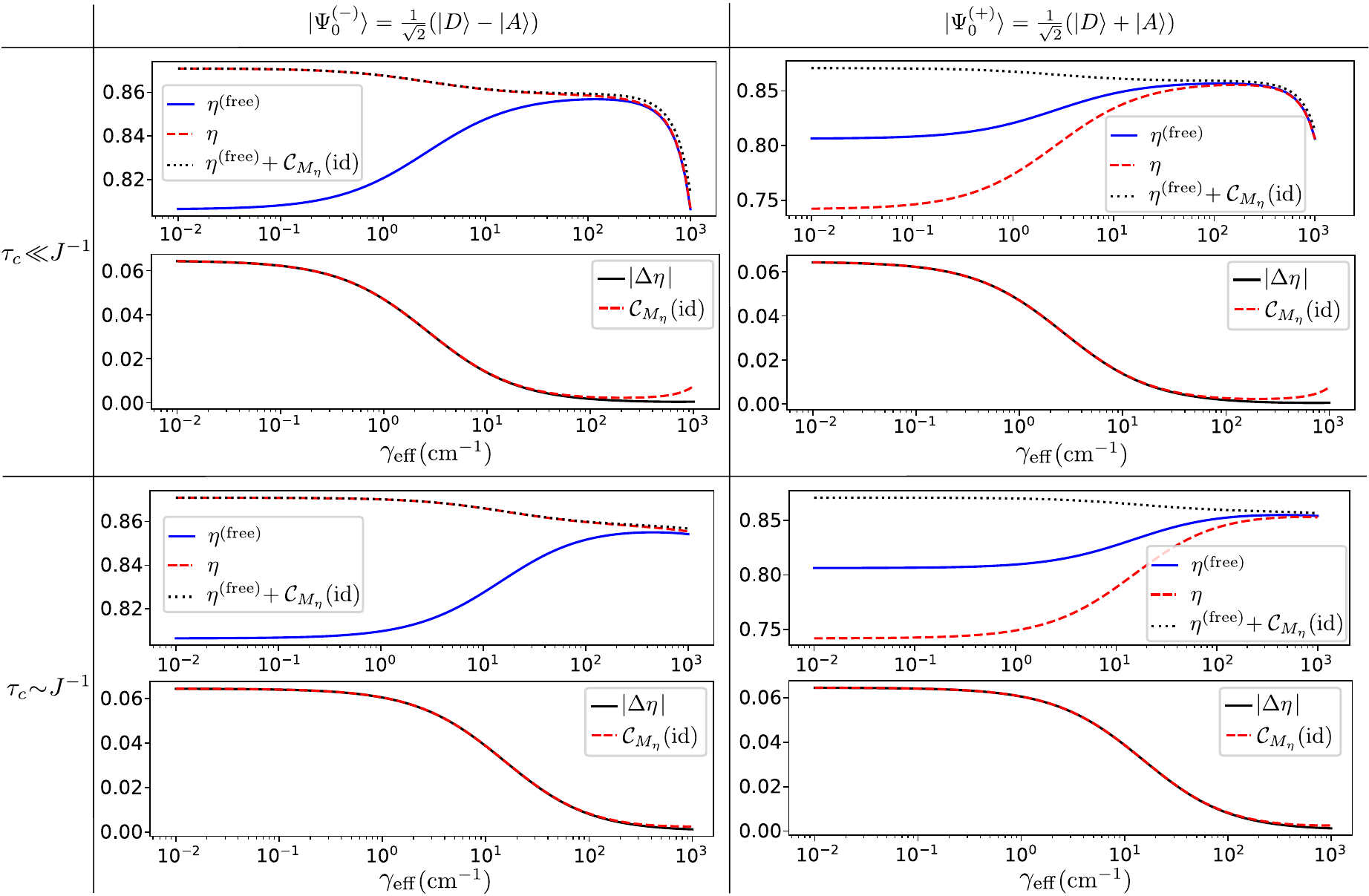}
\caption{Bounds on the coherence-induced change in trapping efficiency for the donor-acceptor dimer. 
For each initial state $\ket{\Psi_0^{(\pm)}}=\tfrac{1}{\sqrt{2}}(\ket{D}\pm\ket{A})$ and time scale, the upper panels show the efficiency $\eta$, its dephased baseline $\eta^{(\mathrm{free})}=\eta(\mathcal{G}(\rho_0))$, and the bound $\eta^{(\mathrm{free})}+\mathcal{C}_{M_\eta}(\mathrm{id})$. The lower panels illustrates Eq.~\eqref{eq:DAM-bound-eta-heom}. 
The system parameters are $J=\Delta=100\,\mathrm{cm}^{-1}$, $T=300\,\mathrm{K}$, asymmetric trapping $\kappa_A=1\,\mathrm{ps}^{-1}$, and recombination $\Gamma_A=\Gamma_D=0.1\,\mathrm{ps}^{-1}$. 
Results are shown for a fast bath ($\tau_c=2.65\,\mathrm{fs}$, $\tau_c\ll J^{-1}$; top row) and an intermediate bath ($\tau_c=53.1\,\mathrm{fs}$, $\tau_c\sim J^{-1}$; bottom row). The horizontal axis is the effective noise scale $\gamma_{\mathrm{eff}}$ obtained from the fast-bath/high-$T$ mapping. In the intermediate-bath regime it is used only as a monotonic parametrization of $E_R$ at fixed $\tau_c$ and $T$ (see text).\label{fig:DAM-efficiency}}
\end{figure*}

\subsection{Bounds on efficiency and average transfer time\label{sec:efficiency-DAM}}

In the following, we connect the resource impact functional $\mathcal{C}_M(\Lambda)$ to standard figures of merit for excitation energy transfer, namely the trapping efficiency and the average energy transfer time (e.g., see Ref.~\cite{RMKLAG09}). 
Throughout, we consider the reduced dynamics on the donor-acceptor manifold (in the single-excitation subspace) and describe environmental decoherence using HEOM as in Sec.~\ref{sec:DAM-int}. 

The trapping efficiency for an initial state $\rho_0$ is given by 
\begin{equation}\label{eq:efficiency-heom}
\eta(\rho_0)=\int_0^\infty\! \mathrm dt\, \Tr[\widetilde M\,\rho(t)]=\Tr[M_\eta\,\rho_0]\,,
\end{equation}
where $\widetilde M:=\sum_{j=D,A}\kappa_j\ket{j}\!\bra{j}$ and we have the effective measurement operator
\begin{equation}\label{eq:M-eta-heom}
M_\eta:=\int_0^\infty\!\mathrm dt\,\Lambda_t^\dagger(\widetilde M)\,.
\end{equation}
For a stable Lindblad semigroup $\Lambda_t=\exp(t\mathcal L)$ one has $M_\eta =(-\mathcal L^\dagger)^{-1}(\widetilde M)$ on the subspace where $\mathcal{L}^\dagger$ is invertible.
In contrast, under HEOM the reduced dynamics is in general not a semigroup and need not admit a time-independent Lindblad generator on the donor-acceptor manifold \cite{TK89, IF09-JCP, IS20}. 
Nevertheless, $M_\eta$ can be evaluated without time integration by working in the (time-local) extended HEOM space by writing the HEOM equations as $\mathrm{d}\mathbf R_t/\mathrm{d}t=\mathcal L_{\mathrm{heom}}\,\mathbf R_t$ for the hierarchy vector $\mathbf R_t$ (see Sec \ref{sec:heom-CM} for more details), such that the Lindblad generator $\mathcal{L}$ is effectively replaced by $\mathcal{L}_{\text{heom}}$. 

Similarly, the unnormalised first moment of the trapping-time distribution is
\begin{equation}
N_\tau(\rho_0):=\int_0^\infty\!\mathrm dt\, t\,\Tr[\widetilde M\,\rho(t)]=\Tr[M_\tau\,\rho_0]\,,
\end{equation}
with effective measurement operator
\begin{equation}\label{eq:M-tau-heom}
M_\tau:=\int_0^\infty\!\mathrm dt\,t\,\Lambda_t^\dagger(\widetilde M)\,.
\end{equation}
In the Lindblad semigroup case $M_\tau = (-\mathcal L^\dagger)^{-2}(\widetilde M)$, whereas in HEOM $M_\tau$ is again obtained from a linear solve in the extended HEOM space. 
The average transfer time is then defined as \footnote{Ref.~\cite{RMKLAG09} formulates trapping and recombination via anti-Hermitian contributions to an effective Hamiltonian. In a Lindblad formulation the corresponding rates differ by a conventional factor of two, which rescales $\eta$ but cancels in $\tau=N_\tau/\eta$.}
\begin{equation}
\tau(\rho_0):=\frac{N_\tau(\rho_0)}{\eta(\rho_0)}\,.
\end{equation}
To quantify the impact of initial coherence, we introduce free baselines by applying the resource-destroying dephasing map $\mathcal{G}$ to the initial state,
\begin{equation}\label{eq:efficiency-free-heom}
\eta^{(\mathrm{free})}(\rho_0):=\eta(\mathcal G(\rho_0))\,,\quad
\tau^{(\mathrm{free})}(\rho_0):=\tau(\mathcal G(\rho_0))\,,
\end{equation}
and set $\Delta\eta:=\eta-\eta^{(\mathrm{free})}$ and $\Delta\tau:=\tau-\tau^{(\mathrm{free})}$. 
Since $\Delta\eta(\rho_0)=\Tr[M_\eta(\rho_0-\mathcal G(\rho_0))]$, the state-specific change in efficiency is bounded by the resource impact functional,
\begin{equation}\label{eq:DAM-bound-eta-heom}
|\Delta\eta(\rho_0)|\leq \mathcal C_{M_\eta}(\mathrm{id})\,,
\end{equation}
and analogously 
\begin{equation}\label{eq:DAM-bound-N}
|\Delta N_\tau(\rho_0)|\le \mathcal C_{M_\tau}(\mathrm{id})\,.
\end{equation}
Moreover, whenever the time-integrated map $\Lambda^\prime := \int_{0}^{\infty} \mathrm{d}t\, \Lambda_t$ is well defined, the pull-back identity $\mathcal{C}_M(\Lambda_2\circ\Lambda_1) = \mathcal{C}_{\Lambda_2^\dagger(M)}(\Lambda_1)$ \cite{LS26-RT} applied with $\Lambda_1=\mathrm{id}$ and $\Lambda_2=\Lambda^\prime$ yields
\begin{equation}\label{eq:DAM-pullback}
\mathcal{C}_{M_\eta}(\mathrm{id}) = \mathcal{C}_{\widetilde{M}}(\Lambda^\prime)\,,
\end{equation}
i.e., $\mathcal{C}_{M\eta}(\mathrm{id})$ can be viewed as the impact functional of $\widetilde{M}$ under the integrated dynamics $\Lambda^\prime$. 

If the free baseline efficiency satisfies $\eta^{(\mathrm{free})}(\rho_0)>\mathcal{C}_{M_\eta}(\mathrm{id})$, then the ratio is well-conditioned and the coherence-induced change in transfer time satisfies
\begin{equation}\label{eq:Dtau-bound-improved-heom}
|\Delta\tau(\rho_0)|\leq \frac{\mathcal{C}_{M_\tau}(\mathrm{id})+\tau^{(\mathrm{free})}(\rho_0)\mathcal{C}_{M_\eta}(\mathrm{id})}{\eta^{(\mathrm{free})}(\rho_0)-\mathcal{C}_{M_\eta}(\mathrm{id})}\,.
\end{equation}

We illustrate the efficiency $\eta$ \eqref{eq:efficiency-heom}, its coherence-free counterpart $\eta^{(\mathrm{free})}$ \eqref{eq:efficiency-free-heom}, and the corresponding upper bound given by the impact functional $\mathcal{C}_{M_\eta}(\mathrm{id})$ in terms of the effective measurement operator $M_\eta$ in Fig.~\ref{fig:DAM-efficiency}.
The environmental noise strength is parametrised by an effective scale $\gamma_{\mathrm{eff}}$ obtained from the fast-bath/high-$T$ mapping of a Drude-Lorentz bath (with $\omega_c=\tau_c^{-1}$), for which $\gamma_{\mathrm{eff}}\approx 4 E_R k_B T/\omega_c$. 
While this identification as a Markovian dephasing rate is controlled only for $\tau_c\ll J^{-1}$, we use $\gamma_{\mathrm{eff}}$ in the intermediate regime $\tau_c\sim J^{-1}$ as a monotonic parametrisation of $E_R$ at fixed $\tau_c$ and $T$.
The parameters are $J=\Delta=100\,\mathrm{cm}^{-1}$, $T=300\,\mathrm{K}$, asymmetric trapping $\kappa_A=1\,\mathrm{ps}^{-1}$, and recombination $\Gamma_A=\Gamma_D=0.1\,\mathrm{ps}^{-1}$. 
We consider two coherent initial states $\rho_0^{(\pm)}=\ket{\Psi_0^{(\pm)}}\!\bra{\Psi_0^{(\pm)}}$ with $\ket{\Psi_0^{(\pm)}}=\tfrac{1}{\sqrt{2}}(\ket{D}\pm\ket{A})$, which probe constructive and destructive interference in transport. 
For each bath timescale (fast bath with $\tau_c=2.65\,\mathrm{fs}$, and intermediate bath with $\tau_c=53.1\,\mathrm{fs}$), the upper panels show $\eta(\rho_0^{(\pm)})$, its baseline $\eta^{(\mathrm{free})}(\rho_0^{(\pm)})$, and the upper bound $\eta^{(\mathrm{free})}(\rho_0^{(\pm)})+\mathcal{C}_{M_\eta}(\mathrm{id})$, whereas the lower panels show $|\Delta\eta(\rho_0^{(\pm)})|$ together with $\mathcal{C}_{M_\eta}(\mathrm{id})$, indicating that the bound is reasonably tight for this preparation and parameter range.

\begin{figure*}[htb]
\centering
\includegraphics[width=\linewidth]{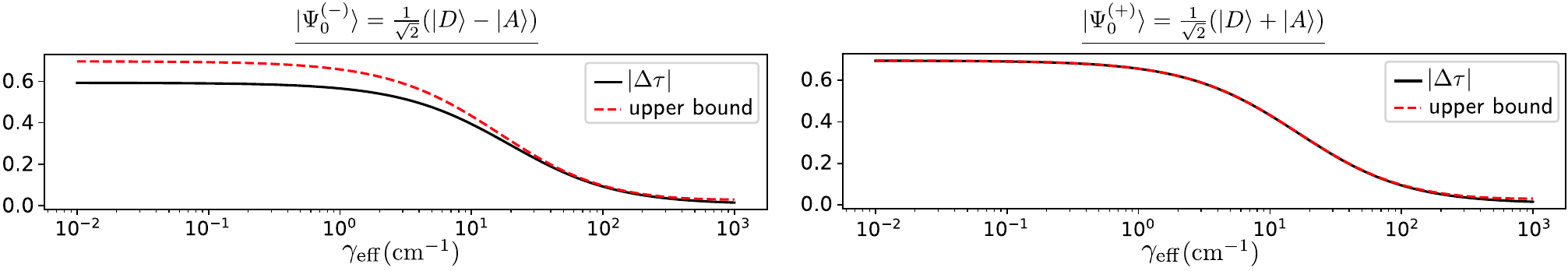}
\caption{Illustration of the state dependence of the upper bound on the coherence induced change $|\Delta\tau|$ in the average energy transfer time in Eq.~\eqref{eq:Dtau-bound-improved-heom} as a function of $\gamma_\mathrm{eff}$ for the donor-acceptor model with $J = \Delta = 100\,\text{cm}^{-1}$, $T= 300\,\mathrm{K}$, $\kappa_A= 1\,\text{ps}^{-1}$, and $\Gamma_A=\Gamma_D = 0.1\,\text{ps}^{-1}$, and the two initial states $\ket{\Psi_0^{(\pm)}} = \tfrac{1}{\sqrt{2}}(\ket{D}\pm\ket{A})$. \label{fig:DAM-tau-bound}}
\end{figure*} 

The coherence-free efficiency $\eta^{(\mathrm{free})}$ exhibits the characteristic environment-assisted quantum transport (ENAQT) behaviour \cite{RMKLAG09}: as the effective noise strength is increased from weak coupling, moderate environmental fluctuations can enhance trapping, whereas sufficiently strong noise suppresses transport again. 
For coherent initial preparations $\rho_0^{(\pm)}$, the efficiencies $\eta(\rho_0^{(+)})$ and $\eta(\rho_0^{(-)})$ differ because the relative phase of the initial superposition can either support or hinder population transfer toward the trap, i.e., dephasing can therefore either reduce a beneficial interference pattern or wash out an initially disadvantageous one.
The new element introduced here is that such preparation-dependent coherence effects are quantified and constrained by the resource impact functional, which measures the sensitivity of a chosen readout to initial-state coherence under the given dynamics (recall $\mathcal{C}_{M_\eta}(\mathrm{id})=\mathcal{C}_{\widetilde{M}}(\Lambda^\prime)$ \eqref{eq:DAM-pullback}).

It is important to note that the tightness of the bounds in Eqs.~\eqref{eq:DAM-bound-eta-heom}, \eqref{eq:DAM-bound-N} and \eqref{eq:Dtau-bound-improved-heom} depends both on the system parameters and on the chosen initial state, i.e., on how the coherence present in $\rho_0$ aligns (in magnitude and phase) with the optimal directions selected by the effective observables $M_\eta$ and $M_\tau$. The prescribed dynamics $\Lambda_t$ enters these bounds only through these effective operators (and, for the transfer-time bound, also through the free $\eta^{(\mathrm{free})}(\rho_0)$ and $\tau^{(\mathrm{free})}(\rho_0)$). 

As an illustration of this preparation dependence, Fig.~\ref{fig:DAM-tau-bound} compares $|\Delta\tau(\rho_0)|$ with the corresponding upper bound in Eq.~\eqref{eq:Dtau-bound-improved-heom} for the same two initial states and system parameters as in the lower row of Fig.~\ref{fig:DAM-efficiency}. While in Fig.~\ref{fig:DAM-efficiency}, the two preparations do not influence the tightness of the bound on $|\Delta\eta|$ in Eq.~\eqref{eq:DAM-bound-eta-heom}, the transfer time bound on $|\Delta\tau|$, Eq.~\eqref{eq:Dtau-bound-improved-heom} is looser for $\ket{\Psi_0^{(-)}}$ whereas it remains comparatively tight over a broad range of $\gamma_\varphi$ for $\ket{\Psi_0^{(+)}}$. At the same time, Figs.~\ref{fig:DAM-efficiency} and \ref{fig:DAM-tau-bound} show that the corresponding upper bounds track the qualitative dependence of $|\Delta\eta|$ and $|\Delta\tau|$ on $\gamma_\varphi$, and can thus serve as practical indicators for when initial state coherence can have an operational impact on $\eta$ and $\tau$.

\subsection{Finite gate times and coarse-graining of measurement operators\label{sec:DAM-CGI}}

In practice, detections are performed over a finite time gate, implying that the relevant observable is not an instantaneous projector but the relevant measurement effect (or POVM) is time-integrated over the detector's acquisition window. In fluorescence or time-correlated single-photon counting on molecular dimers, for instance, the electronics define a finite time window during which emission events are registered. In the donor-acceptor model, this is incorporated directly at the POVM level by replacing the ideal projectors with effective measurement operators obtained by integrating the time-dependent detection POVM over the gate. Because this procedure averages the dynamics within the gate, it cannot be more informative than an idealized projective readout. Thus, this provides a setting to quantify how experimental coarse-graining limits any observable resource-induced advantage.

As a reference point, we consider the projective POVM
\begin{equation}
\{M_A,M_D,M_0\}=\left\{\ket{A}\!\bra{A},\,\ket{D}\!\bra{D},\,\mathbbm{1}-\ket{A}\!\bra{A}-\ket{D}\!\bra{D}\right\},
\end{equation}
where $M_A$ and $M_D$ correspond to population on the acceptor and donor sites, and $M_0=\ket{g}\!\bra{g}$ to the overall ground state. In the following, we will replace this measurement at sharp times by a time-gated photon-counting model in which the system first evolves under the donor-acceptor dynamics and is then measured with a gate-dependent POVM element.

Concretely, let $\rho_0$ be the initial state at $t=0$ and let $\{\Lambda_t\}_{t\ge 0}$ denote the (pre-detection) dynamical map of the donor-acceptor model. We consider a detection gate of length $\Delta t$ that starts at time $t\geq 0$, i.e., the detector integrates click events in the interval $[t,t+\Delta t]$. The state at the beginning of the gate is therefore
\begin{equation}
\rho(t)=\Lambda_t(\rho_0)\,.
\end{equation}
Photon detection in channel $j\in\{A,D\}$ is modelled by a jump operator $L_j$ (e.g., $L_A=\sqrt{\Gamma_A}\ket{g}\!\bra{A}$ and $L_D=\sqrt{\Gamma_D}\ket{g}\!\bra{D}$). The corresponding instantaneous click rate is
\begin{equation}
p_j(t')\,\mathrm{d}t'=\Tr \left[L_j^\dagger L_j\,\rho(t') \right]\,\mathrm{d}t'\,.
\end{equation}
We additionally allow for a possibly time-dependent detection efficiency within the gate, described by weights $0\le w_j(\tau)\le 1$ for $\tau\in[0,\Delta t]$. The total probability to register a click in channel $j$ during the gate is then
\begin{align}
p_j(t,\Delta t) &=\int_0^{\Delta t}\!\mathrm{d}\tau\, w_j(\tau)\,\Tr \bigl[L_j^\dagger L_j\,\rho(t+\tau)\bigr]\nonumber\\
&=\int_0^{\Delta t}\!\mathrm{d}\tau\, w_j(\tau)\,\Tr \bigl[L_j^\dagger L_j\,\Lambda_{t+\tau}(\rho_0)\bigr]\,.
\end{align}
In the Markovian Lindblad setting $\Lambda_{t+\tau}=\Lambda_\tau\circ\Lambda_t$, such that $\rho(t+\tau)=\Lambda_\tau(\rho(t))$ and the gate probability can be written as a measurement acting after the evolution to the gate start time:
\begin{equation}
p_j(t,\Delta t)=\Tr \left[\widetilde{M}_j(\Delta t) \rho(t)\right]\,,
\end{equation}
with the effective time-integrated POVM element
\begin{equation}\label{eq:Mtilde-gate}
\widetilde{M}_j(\Delta t) :=\int_0^{\Delta t}\!\mathrm{d}\tau\, \Lambda_\tau^\dagger \left(L_j^\dagger L_j\right)  w_j(\tau)  \geq 0\,.
\end{equation}

In the single-excitation donor-acceptor subspace without re-excitation, at most one recombination photon is emitted, so the integrated click rate equals the click probability and the gate effects define a valid three-outcome POVM. In particular, since $\sum_{j\in\{A,D\}}p_j(t,\Delta t)\leq 1$ for all input states, one has
\begin{equation}
\sum_{j\in\{A,D\}}\widetilde{M}_j(\Delta t)\leq \mathbbm{1}\,.
\end{equation}
In combination with the third POVM element $\widetilde{M}_0(\Delta t):=\mathbbm{1}-\widetilde{M}_A(\Delta t)-\widetilde{M}_D(\Delta t)$ this yields the gate-dependent POVM $\{\widetilde{M}_A(\Delta t),\widetilde{M}_D(\Delta t),\widetilde{M}_0(\Delta t)\}$.

In this framework, the ideal sharp-time projectors $M_j$ in the definition of the resource impact functional are replaced by the gate-dependent effects $\widetilde{M}_j(\Delta t)$. The quantity $\mathcal{C}_{\widetilde{M}_j(\Delta t)}(\Lambda_t)$ therefore quantifies the maximal resource-induced change in the integrated detection probability in channel $j$ over the gate $[t,t+\Delta t]$, rather than the population on site $A$ or $D$ at a single time. Since $\widetilde{M}_j(\Delta t)=O(\Delta t)$ as $\Delta t\to 0$, also $\mathcal{C}_{\widetilde{M}_j(\Delta t)}(\Lambda_t)\to 0$ in this limit. For sufficiently small gates it is thus natural to consider the rescaled quantity $\mathcal{C}_{\widetilde{M}_j(\Delta t)}(\Lambda_t)/\Delta t$, which converges to an instantaneous advantage rate determined by $L_j^\dagger L_j$ and $w_j(0)$.

Moreover, depending on the dynamics during the gate and on the efficiency profile $w_j(\tau)$, the maximal resource-induced change quantified by $\mathcal{C}_{\widetilde{M}_j(\Delta t)}(\Lambda_t)$ can be strongly suppressed compared to the idealized sharp-time readout $\mathcal{C}_{M_j}(\Lambda_t)$. This highlights that the operational impact of a resource depends not only on the channel $\Lambda_t$ but also on the specific task (yield) encoded by the measurement operator. For example, if $\Lambda_t$ is the dynamical map from a Lindblad or HEOM model, then decoherence due to system-bath interactions and amplitude damping acting during the gate reduce the off-diagonal contribution to $\widetilde{M}_A(\Delta t)$, and hence can reduce $\mathcal{C}_{\widetilde{M}_A(\Delta t)}(\Lambda_t)$ relative to $\mathcal{C}_{M_A}(\Lambda_t)$.

In addition to this temporal coarse-graining, one may further classically coarse-grain the discrete outcomes via a stochastic matrix $V$, for instance by merging the two exciton channels or modelling classical misassignment between $A$ and $D$. The resulting POVM elements $\widehat{M}_k=\sum_j V_{kj}\widetilde{M}_j(\Delta t)$ satisfy the classical data-processing inequality (see Ref.~\cite{LS26-RT})
\begin{equation}
\mathcal{C}_{\widehat{M}_k}(\Lambda_t)\le \sum_j V_{kj}\,\mathcal{C}_{\widetilde{M}_j(\Delta t)}(\Lambda_t)\,,
\end{equation}
implying that additional classical post-processing can only further reduce the observable resource advantage. Taken together, the gate construction and subsequent classical coarse-graining make explicit how both the system--bath evolution and finite detection resolution constrain the extent to which quantum coherence can be operationally witnessed or exploited in experimentally accessible exciton-transport tasks.

\section{Multi-site donor chain with trapping \label{sec:Kette}}

To move beyond the dimer, we consider an $N$-site linear donor chain with terminal trapping.
In contrast to the two-site case, the optimal initial state for trapping is generally non-trivial: maximal trapping can require a trade-off between initial delocalization, that is site-basis coherence, which promotes spreading, and initial population biased towards the terminal site $N$.

For chain lengths up to $N=50$ we model the reduced dynamics using a time-homogeneous GKSL (Lindblad) master equation. While HEOM provides a nonperturbative description of structured environments, its computational cost grows rapidly with system size and with the hierarchy depth required for convergence, and practical implementations are typically limited to comparatively small systems. 
In the weak electronic coupling or short bath memory regime, a Lindblad model captures the competition between coherent hopping, dephasing, recombination, and irreversible trapping that underlies environment-assisted transport schemes \cite{MRLAG08, PH08, CCDHP09, RMKLAG09, KAG12}. Importantly, the theoretical results derived in Sec.~\ref{sec:theory-chain} depend only on the resulting family of dynamical maps $\Lambda_t$ and therefore apply irrespective of whether $\Lambda_t$ is obtained from Lindblad dynamics or HEOM. 

Moreover, we work in the zero- and single-excitation manifold
\begin{equation}\label{eq:H-chain}
\mathcal{H}=\mathrm{span}_{\CC}\{\ket{g},\ket{s},\ket{1},\dots,\ket{N}\}\,,
\end{equation}
where $\ket{n}$ denotes the state with the excitation localized on site $n$, $\ket{g}$ is the global ground state, and $\ket{s}$ is an absorbing sink state. The coherent dynamics within the donor manifold is generated by a tight-binding Hamiltonian
\begin{equation}
H=\sum_{n=1}^N \varepsilon_n \ket{n}\!\bra{n}+\sum_{n=1}^{N-1}J_n\left(\ket{n}\!\bra{n+1}+\ket{n+1}\!\bra{n}\right)\,.
\end{equation}

\begin{figure}[tb]
\centering
\includegraphics[width=\linewidth]{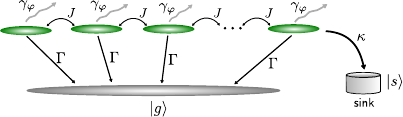}
\caption{Schematic of an $N$-site linear chain in the zero- and single-excitation manifold. An
excitation can be trapped from site $N$ into an absorbing sink $\ket{s}$ at rate $\kappa$, recombine
to the global ground state $\ket{g}$ at rate $\Gamma$, and undergo site-local pure dephasing at rate
$\gamma_{\varphi}$. \label{fig:Kette}}
\end{figure}

Dissipation is modelled by Lindblad jump operators describing:
(i) recombination to the ground state,
\begin{equation}
L^{(\mathrm{rec})}_n=\sqrt{\Gamma}\,\ket{g}\!\bra{n}\,,\quad n=1,\dots,N\,,
\end{equation}
(ii) trapping into the sink from the terminal site,
\begin{equation}
L^{(\mathrm{trap})}=\sqrt{\kappa}\,\ket{s}\!\bra{N}\,,
\end{equation}
and (iii) local pure dephasing in the site basis,
\begin{equation}
L^{(\mathrm{deph})}_n=\sqrt{\gamma_{\varphi}}\,\ket{n}\!\bra{n}\,,\quad n=1,\dots,N\,,
\end{equation}
as illustrated in Fig.~\ref{fig:Kette}. The resulting dynamics defines a CPTP semigroup $\Lambda_t=e^{t\mathcal L}$.
Since we are interested in coherences within the donor manifold and there are no coherent couplings between $\{\ket{n}\}_{n=1}^N$ and $\ket{g},\ket{s}$, we restrict attention to initial preparations supported on the donor subspace,
\begin{equation}\label{eq:D-D}
\mathcal{D}_D:=\left\{\rho\in\mathcal{D}(\mathcal{H})\,\middle\vert\, \rho = \Pi_D \rho\Pi_D\right\}\,,\quad \Pi_D = \sum_{n=1}^N\ket{n}\!\bra{n}\,.
\end{equation}
All optimizations over input states (e.g., in Eq.~\eqref{eq:capacity-sup}) are understood to be
restricted to $\rho\in\mathcal{D}_D$.

As in the previous section, we consider the resource theory of coherence in the site basis. The complete dephasing map reads
$\{\ket{g},\ket{s},\ket{1},\dots,\ket{N}\}$,
\begin{equation}
\mathcal{G}(\rho)=\ket{g}\!\bra{g}\rho\ket{g}\!\bra{g} +\ket{s}\!\bra{s}\rho\ket{s}\!\bra{s} +\sum_{n=1}^N \ket{n}\!\bra{n}\rho\ket{n}\!\bra{n}\,,
\end{equation}
which removes all off-diagonal coherences in the site basis and, when restricted to $\mathcal D_D$,
reduces to complete dephasing on the donor manifold.

\subsection{Time-dependent trapping yield and $\mathcal{C}_M(\Lambda_t)$ \label{sec:CM-chain}}

\begin{figure*}[tb]
\centering
\includegraphics[width=\linewidth]{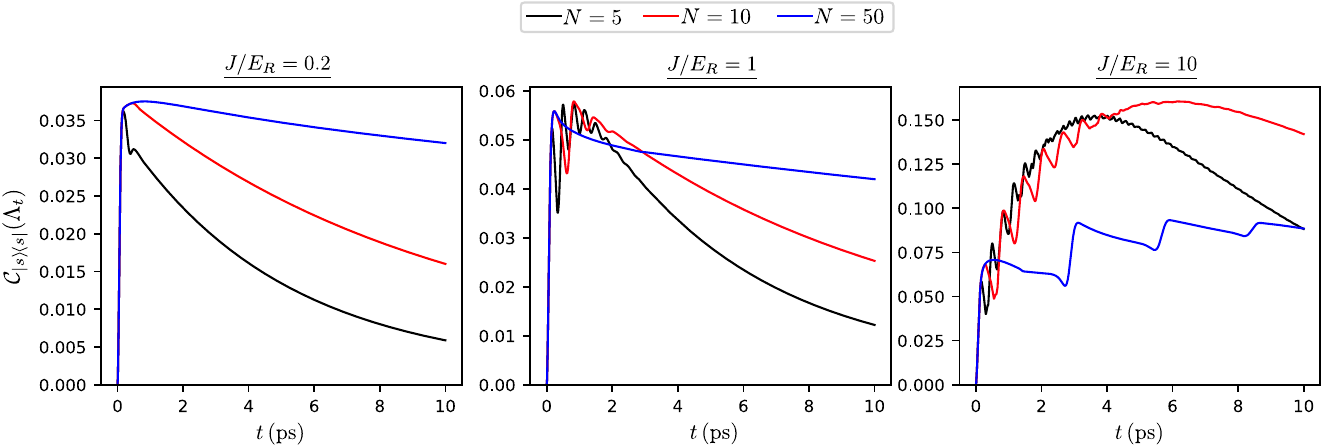}
\caption{Time-resolved resource impact functional $\mathcal{C}_{|s\rangle\!\langle{s}|}(\Lambda_t)$ for a homogeneous one-dimensional $N$-site chain with an absorbing sink attached to site $N$, for $N\in\{5,10,50\}$ and coupling ratios $J/E_R\in\{0.2,1,10\}$. The fixed system parameters are $\varepsilon_n=0$, $J_n=J=100\,\mathrm{cm}^{-1}$, $T=300\,\mathrm{K}$, $\Gamma=0.01\,\mathrm{ps}^{-1}$, $\kappa=1\,\mathrm{ps}^{-1}$, and $\tau_c=0.01\,J^{-1}$ corresponding to a Markovian bath. A value of $\mathcal{C}_{|s\rangle\!\langle s|}(\Lambda_t)= 0.1$ means that that initial coherence can modify the yield at time $t$ by at most $10$ percentage points relative to its dephased reference. \label{fig:CM-chain}}
\end{figure*}

For initial states supported on the donor manifold, the cumulative amount of excitation irreversibly transferred to the sink up to time $t$ is quantified by the time-dependent trapping yield
\begin{equation}\label{eq:eta-trap-t}
\eta_{\mathrm{trap}}(t;\rho_0):= \kappa   \int_0^t \mathrm{d}t^\prime\,\bra{N}\rho(t^\prime)\ket{N}   = \Tr\big[M_{\mathrm{trap}}(t)\rho_0\big]
\end{equation}
with the effective (time-dependent) measurement operator
\begin{equation}\label{eq:Mtrap-t}
M_{\mathrm{trap}}(t) := \kappa \int_0^t \mathrm{d}t^\prime\,\Lambda_{t^\prime}^\dagger \big(\ket{N}\!\bra{N}\big)\,.
\end{equation}
For the irreversible sink jump operator $L^{(\mathrm{trap})}=\sqrt{\kappa}\ket{s}\!\bra{N}$ and no
reverse process, this yield coincides with the sink population $\bra{s}\rho_t\ket{s}$. 
The long-time limit defines the trapping efficiency
\begin{align}\label{eq:eta-max-chain}
\eta_{\mathrm{trap}}(\rho_0) &:= \lim_{t\to\infty}\eta_{\mathrm{trap}}(t;\rho_0)= \Tr \big[M_{\mathrm{trap}}\,\rho_0\big]\,,\\
M_{\mathrm{trap}}&:=\lim_{t\to\infty}M_{\mathrm{trap}}(t)\,,
\end{align}
in direct analogy to Eq.~\eqref{eq:efficiency-heom}. Note that in the presence of recombination, $\eta_{\mathrm{trap}}(\rho_0)\le 1$ in general.

To isolate the role of initial site coherence, we compare $\eta_{\mathrm{trap}}(t;\rho_0)$ with its
coherence-free counterpart obtained by dephasing the input state,
\begin{align}
\eta_{\mathrm{trap}}^{(\mathrm{free})}(t;\rho_0)&:=\Tr \big[M_{\mathrm{trap}}(t)\,\mathcal{G}(\rho_0)\big]\,,\\
\Delta\eta_{\mathrm{trap}}(t;\rho_0)&:=\eta_{\mathrm{trap}}(t;\rho_0)-\eta_{\mathrm{trap}}^{(\mathrm{free})}(t;\rho_0)\,,
\end{align}
and analogously at $t\to\infty$.
Since $\mathcal{G}$ is linear and Hermitian, the absolute coherence-induced change is bounded by the resource impact functional according to 
\begin{align}\label{eq:Delta-eta-chain}
\big|\Delta\eta_{\mathrm{trap}}(t;\rho_0)\big| &\leq \mathcal{C}_{M_{\mathrm{trap}}(t)}(\mathrm{id})=\mathcal{C}_{|s\rangle\!\langle s|}(\Lambda_t)\,,\\
\big|\Delta\eta_{\mathrm{trap}}(\rho_0)\big| &\leq \mathcal{C}_{M_{\mathrm{trap}}}(\mathrm{id})  \,.
\end{align}
The equality $\mathcal{C}_{M_{\mathrm{trap}}(t)}(\mathrm{id})=\mathcal{C}_{|s\rangle\!\langle s|}(\Lambda_t)$
is a specific feature of the model considered in this section:
in the Heisenberg picture, $\frac{\mathrm{d}}{\mathrm{d}t}\Lambda_t^\dagger(\ket{s}\!\bra{s}) =\kappa\,\Lambda_t^\dagger(\ket{N}\!\bra{N})$, such that $\Lambda_t^\dagger(\ket{s}\!\bra{s})=\ket{s}\!\bra{s}+M_{\mathrm{trap}}(t)$.

Following Ref.~\cite{RMKLAG09}, a Markovian pure-dephasing rate $\gamma_\varphi$ from an Ohmic spectral density with cutoff $\omega_c$ is used. 
In the Markovian regime the dephasing rate is determined by the zero-frequency slope,
\begin{equation}\label{eq:gamma_phi_ohmic}
\gamma_\varphi(T)=\frac{2\pi k_B T}{\hbar}\left.\frac{ \partial J (\omega)}{\partial\omega}\right|_{\omega=0}\,, 
\end{equation}
and for an Ohmic form with exponential cutoff this yields
$\gamma_\varphi(T)=2\pi (k_B T/\hbar)\,E_R/(\hbar\omega_c)$ \cite{RMKLAG09}.
With $T$ and $\omega_c$ fixed, varying $E_R$ therefore directly tunes the dimensionless noise strength
$\gamma_\varphi/J\propto E_R/J$. 

Figure \ref{fig:CM-chain} shows $\mathcal{C}_{|s \rangle\!\langle s|}(\Lambda_t)$, which provides a state-independent upper bound on the magnitude of any coherence-induced change $|\Delta\eta_{\mathrm{trap}}(t;\rho_0)|$ at time $t$.
Since $M=\ket{s}\!\bra{s}$ is a projector and thus satisfies $\|M\|_\infty=1$, the quantity
$Y_t(\rho_0):=\Tr\big[M \Lambda_t(\rho_0)\big]\in[0,1]$ can be interpreted as the probability that the excitation has been trapped in the sink by time $t$.
By construction, $ |\Delta Y_t(\rho_0)|:=|Y_t(\rho_0) - Y(\mathcal{G}(\rho_0))| \leq \mathcal{C}_{|s\rangle\!\langle s|}(\Lambda_t)$, such that, for example, $\mathcal C_{|s\rangle\!\langle s|}(\Lambda_t)=0.1$ means that initial coherence can modify the yield at time $t$ by at most $10$ percentage points relative to its dephased reference. 
Moreover, the utilization ratio $|\Delta Y_t(\rho_0)|/\mathcal{C}_{|s\rangle\!\langle s|}(\Lambda_t)\leq 1$ quantifies how closely a given preparation approaches this bound (see also Sec.~\ref{sec:DAM-int}).

In the strong-noise regime (small $J/E_R$, i.e., large $\gamma_\varphi/J$), the functional rises at very short times and then decays essentially monotonically: phase sensitivity is confined to an early-time window and becomes negligible at later times as dephasing and loss suppress the influence of initial phase relations. For weaker noise (larger $J/E_R$), $\mathcal{C}_{|s \rangle\!\langle s|}(\Lambda_t)$ exhibits pronounced oscillations at early and intermediate times, reflecting coherent propagation and finite-size recurrences of tight-binding dynamics with trapping at the final site \cite{PH08, ZCBK17}.

In the strong electronic coupling regime (right panel), the smaller early time values and the slower envelope modulations observed for larger $N$ can be understood as finite-size propagation effects. In the Heisenberg picture, the sink readout $\Lambda_t^\dagger(\ket{s}\!\bra{s})$ develops off-diagonal contributions on more distant sites only after the Hamiltonian has had time to couple the terminal site to the rest of the chain. Equivalently, in the Schrödinger picture, population must propagate from an initially localised excitation near site $1$ to the terminal site $N$ before trapping. For a homogeneous tight-binding chain this propagation is approximately ballistic in the near-coherent limit, with a characteristic arrival time scaling as $t_{\mathrm{arr}}\sim (N-1)/J$ (up to an $O(1)$ prefactor), which delays the build-up of phase sensitivity and leads to longer recurrence times as $N$ increases.

\subsection{Relating delocalization and trapping efficiency \label{sec:theory-chain}}

The previous subsection used $\mathcal{C}_{|s\rangle\!\langle s|}(\Lambda_t)$ as a time-resolved, state-independent diagnostic of how sensitive the trapping readout can be to site-basis coherence. We now translate this diagnostic into quantitative performance statements for the trapping task.
The main ingredient is that the trapping yield at time $t$ is a linear functional of the initial donor state, where all dynamical information is contained in the positive effect operator $M_{\mathrm{trap}}(t)$.
Its diagonal part in the site basis determines the best incoherent (site-localized) preparation, whereas the off-diagonal part encodes all phase sensitivity. For complete dephasing $\mathcal{G}$, this off-diagonal part is precisely $(\mathrm{id}-\mathcal{G})M_{\mathrm{trap}}(t)$ and its operator norm equals the resource impact functional $\mathcal{C}_{|s\rangle\!\langle s|}(\Lambda_t)$.
Consequently, $\mathcal{C}_{|s\rangle\!\langle s|}(\Lambda_t)$ sets the fundamental scale for any coherent advantage and allows us to (i) bound the maximal coherent-incoherent performance gap, (ii) quantify the minimum delocalization required to achieve a given improvement, and (iii) identify regimes where the trapping-optimal state must remain close to a localized site state.

\subsubsection{Bounds on optimal coherent and incoherent trapping}

We first define the maximal achievable trapped amount at time $t$, optimized over input states supported on the donor manifold,
\begin{equation}\label{eq:eta-max-t}
\eta_{\mathrm{trap}}^{(\max)}(t) := \max_{\rho\in\mathcal{D}_D}\Tr \big[M_{\mathrm{trap}}(t)\rho\big]\,,
\end{equation}
where $M_{\mathrm{trap}}(t)$ is given in Eq.~\eqref{eq:Mtrap-t} and $\mathcal{D}_D$ in Eq.~\eqref{eq:D-D}. 
Similarly, we define the maximal trapped amount achievable by incoherent preparations,
\begin{equation}\label{eq:eta-incoh-max-t}
\eta_{\mathrm{trap}}^{(\mathrm{incoh})}(t):= \max_{\substack{\rho\in\mathcal{D}_D\\ \rho=\mathcal{G}(\rho)}}\Tr \big[M_{\mathrm{trap}}(t)\rho\big]\,,
\end{equation}
such that by definition $\eta_{\mathrm{trap}}^{(\mathrm{incoh})}(t)\leq \eta_{\mathrm{trap}}^{(\max)}(t)$ for all $t\geq 0$.   

To make the role of coherence explicit, we decompose $M_{\mathrm{trap}}(t)=D(t)+B(t)$ with
\begin{equation}\label{eq:DB-decomp}
D(t) := \mathcal{G} \big(M_{\mathrm{trap}}(t)\big)\,,\quad
B(t) :=  (\mathrm{id}-\mathcal{G}) \big(M_{\mathrm{trap}}(t)\big)\,,
\end{equation}
where $D(t)$ is diagonal in the site basis and $B(t)$ contains the off-diagonal contributions responsible for phase sensitivity of the readout. Moreover, $\|B(t)\|_\infty=\mathcal{C}_{M_{\mathrm{trap}}(t)}(\mathrm{id})$, and for an absorbing sink one has $\mathcal{C}_{M_{\mathrm{trap}}(t)}(\mathrm{id})=\mathcal{C}_{|s\rangle\!\langle{s}|}(\Lambda_t)$ as explained in the previous section.

\begin{thm}[Coherence advantage gap bound for trapping]\label{thm:eta-max-coh}
For each $t\ge 0$, let $M_{\mathrm{trap}}(t)$ be defined as in Eq.~\eqref{eq:Mtrap-t}, and let $\eta_{\mathrm{trap}}^{(\max)}(t)$ and $\eta_{\mathrm{trap}}^{(\mathrm{incoh})}(t)$ be given by Eqs.~\eqref{eq:eta-max-t} and \eqref{eq:eta-incoh-max-t}. Then, 
\begin{equation}\label{eq:eta-gap-bound}
\eta_{\mathrm{trap}}^{(\max)}(t)  -\eta_{\mathrm{trap}}^{(\mathrm{incoh})}(t)   \leq \mathcal{C}_{M_{\mathrm{trap}}(t)}(\mathrm{id})\,.
\end{equation}
\end{thm}
The proof of Theorem \ref{thm:eta-max-coh} is given in Appendix \ref{app:thm-eta-max-coh}, and is a consequence of $\mathcal{C}_{M_{\mathrm{trap}}(t)}(\mathrm{id}) = \mathcal{C}_{|s\rangle\!\langle s|}(\Lambda_t)$ (see also Corollary 2 in Ref.~\cite{LS26-RT}). 
In particular, Eq.~\eqref{eq:eta-gap-bound} shows that $\mathcal{C}_{M_{\mathrm{trap}}(t)}(\mathrm{id})$ controls the largest possible performance gap between coherent and incoherent initial preparations at time $t$.
In particular, if $\mathcal{C}_{M_{\mathrm{trap}}(t)}(\mathrm{id})$ is small, then delocalization cannot be functionally important for optimizing the trapped amount at that time.
This is complementary to the state-wise bound in Eq.~\eqref{eq:Delta-eta-chain}, which controls $|\Delta\eta_{\mathrm{trap}}(t;\rho_0)|$ for a fixed preparation $\rho_0$ relative to its dephased counterpart $\mathcal{G}(\rho_0)$. In contrast, Eq.~\eqref{eq:eta-gap-bound} compares two separately optimized benchmarks, $\eta_{\mathrm{trap}}^{(\max)}(t)$ and $\eta_{\mathrm{trap}}^{(\mathrm{incoh})}(t)$, and therefore quantifies the maximal performance gap between coherent and incoherent initial preparations at time $t$ without considering a specific pairing $\rho\mapsto\mathcal{G}(\rho)$.

\begin{cor}[Pairwise lower bound and sufficient condition]\label{cor:2site_lb}
Let $t\geq 0$ be fixed and $M_{\mathrm{trap}}(t)=D(t)+B(t)$ as in Eq.~\eqref{eq:DB-decomp}, with $D(t)$
diagonal in the site basis and $B(t)$ purely off-diagonal. For any $n\neq m$ define
\begin{equation}
\lambda_{nm}^{(+)}(t):= \frac{D_{nn}(t)+D_{mm}(t)}{2} +\sqrt{\left(\frac{\Delta_{nm}(t)}{2}\right)^2+|B_{nm}(t)|^2}.
\end{equation}
where $\Delta_{nm}(t):=D_{nn}(t)-D_{mm}(t)$.
Then,
\begin{equation}\label{eq:2site_lb_main}
\eta_{\mathrm{trap}}^{(\max)}(t)\ \ge\ \lambda_{nm}^{(+)}(t).
\end{equation}
Moreover, if $n$ attains the incoherent optimum, i.e., $D_{nn}(t)=\eta_{\mathrm{trap}}^{(\mathrm{incoh})}(t)$,
then
\begin{align}\label{eq:2site_adv_lb}
\eta_{\mathrm{trap}}^{(\max)}(t)-\eta_{\mathrm{trap}}^{(\mathrm{incoh})}(t) &\geq\sqrt{\left(\frac{\Delta_{nm}(t)}{2}\right)^2+|B_{nm}(t)|^2}\nonumber\\
&\quad -\frac{|\Delta_{nm}(t)|}{2}\,,
\end{align}
\end{cor}
The proof of Corollary \ref{cor:2site_lb} follows from standard matrix analysis and is provided in Appendix~\ref{app:cor-2site_lb}. Beyond providing a computable lower bound on $\eta_{\mathrm{trap}}^{(\max)}(t)$, the corollary is practically useful because it yields a pairwise screening criterion for when site-basis coherence can be functionally relevant for trapping at time $t$. The diagonal entries $D_{nn}(t)$ quantify the trapped amount attainable from
incoherent (site-localized) preparations $\ket{n}$, while the off-diagonal entry $B_{nm}(t)$ encodes the phase sensitivity between sites $n$ and $m$. Corollary~\ref{cor:2site_lb} shows that a sizeable coherent advantage requires that there are competing incoherent strategies (i.e., sites with nearly equal values of $D_{nn}(t)$) and simultaneously a non-negligible interference scale $|B_{nm}(t)|$. 

To make this quantitative, assume without loss of generality that $D_{nn}(t)\geq D_{mm}(t)$ and define
$\Delta_{nm}(t):=D_{nn}(t)-D_{mm}(t)\ge 0$ and $b_{nm}(t):=|B_{nm}(t)|$. The advantage lower bound
\eqref{eq:2site_adv_lb} can then be written as $\eta_{\mathrm{trap}}^{(\max)}(t)-\eta_{\mathrm{trap}}^{(\mathrm{incoh})}(t)
\geq f(\Delta_{nm}(t),b_{nm}(t))$ with $f(\Delta,b):=\sqrt{b^2+\frac{\Delta^2}{4}}-\frac{\Delta}{2}$.
In the near-degenerate regime $\Delta\le 2b$, set $z:=\Delta^2/(4b^2)\in[0,1]$. Using the binomial series $(1+z)^{1/2}=1+\tfrac12 z-\tfrac18 z^2+R_3(z)$ with remainder satisfying $0\leq R_3(z)\leq \tfrac1{16}z^3$ for $z\in[0,1]$, we obtain the expansion
\begin{equation}\label{eq:adv_smallgap}
f(\Delta,b)= b-\frac{\Delta}{2}+\frac{\Delta^2}{8b}-\frac{\Delta^4}{128\,b^3}+R_{\mathrm{nd}}(\Delta,b)\,,
\end{equation}
with $0\leq R_{\mathrm{nd}}(\Delta,b)\le \frac{\Delta^6}{1024\,b^5}$.
In particular, for $\Delta\ll b$ the guaranteed advantage scales linearly with the coherence sensitivity $b_{nm}(t)=|B_{nm}(t)|$, up to a small correction set by the diagonal splitting $\Delta_{nm}(t)$.

Conversely, in the well-separated regime $\Delta\ge 2b$, we set $\varepsilon:=4b^2/\Delta^2\in[0,1]$ and expand $\sqrt{1+\varepsilon}$ to obtain
\begin{equation}\label{eq:adv_largegap}
f(\Delta,b)=\frac{b^2}{\Delta}-\frac{b^4}{\Delta^3}+R_{\mathrm{ws}}(\Delta,b)\,,
\end{equation}
where $0\leq R_{\mathrm{ws}}(\Delta,b)\leq \frac{2b^6}{\Delta^5}$. 
Thus, when a single incoherent strategy dominates (large diagonal gap), coherence can improve the readout only perturbatively, with a leading contribution quadratic in $|B_{nm}(t)|$ and suppressed by $1/\Delta_{nm}(t)$.

\subsubsection{Delocalization}
To quantify delocalization, we use the $\ell^1$-coherence of a pure state $\ket{\Psi}=\sum_{n=1}^N \psi_n\ket{n}$ with $\sum_n|\psi_n|^2=1$~\cite{BCP14},
\begin{equation}\label{eq:l1-coherence}
C_{\ell^1}(\ket{\Psi}\!\bra{\Psi}):= \sum_{n\neq m}|\psi_n\psi_m| = \|\Psi\|_1^2-1\,,
\end{equation}
which satisfies $C_{\ell^1}\in [0,N-1]$ (with $C_{\ell^1}=0$ for a site basis state and $C_{\ell^1}=N-1$ for a uniform superposition). Moreover, we define
\begin{equation}\label{eq:cmax}
c_{\max}(t) := \max_{n\neq m}\left|B_{nm}(t)\right| \,,
\end{equation}
with $B(t)$ as defined in Eq.~\eqref{eq:DB-decomp}.
Then for any $\ket{\Psi}$ one has the estimate
\begin{equation}\label{eq:l1-bound}
\left|\bra{\Psi}B(t)\ket{\Psi}\right|\leq c_{\max}(t)\,C_{\ell^1}(\ket{\Psi}\!\bra{\Psi})\,.
\end{equation}
\begin{thm}[Minimal required initial delocalization]\label{thm:min-deloc}
Fix $t\geq 0$. If a pure donor state $\ket{\Psi}$ achieves an improvement of $\delta\geq 0$ over the best incoherent value, i.e.,
\begin{equation}
\bra{\Psi}M_{\mathrm{trap}}(t)\ket{\Psi} \geq \eta_{\mathrm{trap}}^{(\mathrm{incoh})}(t) +\delta\,,
\end{equation}
with $\eta_{\mathrm{trap}}^{(\mathrm{incoh})}(t)$ given by \eqref{eq:eta-incoh-max-t}, then necessarily
\begin{equation}\label{eq:l1-bound-CM}
C_{\ell^1}(\ket{\Psi}\!\bra{\Psi})  \geq \frac{\delta}{c_{\max}(t)} \geq  \frac{\delta}{\mathcal{C}_{M_{\mathrm{trap}}(t)}(\mathrm{id})}\,.
\end{equation}
Moreover, $\ket{\Psi}$ must have support on at least
\begin{equation}\label{eq:m-sites}
m \geq 1+\frac{\delta}{c_{\max}(t)}
\end{equation}
sites.
\end{thm}
We prove Theorem \ref{thm:min-deloc} in Appendix \ref{app:thm-min-deloc}. 
In particular, Theorem~\ref{thm:min-deloc} provides a necessary condition for achieving an improvement $\delta$ in the trapped amount $\eta_{\mathrm{trap}}(t)$ over the best incoherent preparation, and relates this improvement to $\mathcal{C}_{M_{\mathrm{trap}}(t)}(\mathrm{id})$. To outperform the best incoherent preparation by an amount $\delta$, the input must contain at least $\delta/c_{\max}(t)$ of site-basis coherence (as measured by $C_{\ell^1}$), and therefore has to be delocalized over at least $m\ge 1+\delta/c_{\max}(t)$ sites. In this sense, $c_{\max}(t)$ quantifies the largest ``coherence sensitivity'' of the trapping task at time $t$, i.e., the maximal change in $\eta_{\mathrm{trap}}(t)$ that can be produced per unit of $\ell^1$-coherence in the initial state. Moreover, since $c_{\max}(t)\le \mathcal{C}_{M_{\mathrm{trap}}(t)}(\mathrm{id})$, one obtains the simpler but more conservative condition $C_{\ell^1}\ge \delta/\mathcal{C}_{M_{\mathrm{trap}}(t)}(\mathrm{id})$. In fact, if $c_{\max}(t)=0$ (equivalently $(\mathrm{id}-\mathcal{G})M_{\mathrm{trap}}(t)$ is diagonal), then no coherent preparation can improve $\eta_{\mathrm{trap}}(t)$ over the incoherent optimum at that time.

While the $\ell^1$-coherence is particularly convenient here because it matches the entry-wise structure of the off-diagonal contribution $(\mathrm{id}-\mathcal{G})\!\big(M_{\mathrm{trap}}(t)\big)$, an analogue of Eq.~\eqref{eq:l1-bound-CM} in terms of the relative entropy of coherence $C_{\mathrm{rel}}(\rho)=D\!\big(\rho\|\mathcal{G}(\rho)\big)$ \cite{CG19} is given by
\begin{equation}\label{eq:Crel-CM-bound}
C_{\mathrm{rel}}(\rho)\geq \frac{1}{2\ln 2} \left(\frac{\delta}{\mathcal{C}_{M_{\mathrm{trap}}(t)}(\mathrm{id})}\right)^2\,,
\end{equation}
whenever $\Tr[M_{\mathrm{trap}}(t)\rho]\ge \eta_{\mathrm{trap}}^{(\mathrm{incoh})}(t)+\delta$.
This entropic bound applies to arbitrary (pure or mixed) inputs, whereas the additional support-size statement in Eq.~\eqref{eq:m-sites} derived from the $\ell^1$-bound relies on specializing to pure states via $C_{\ell^1}(|\Psi\rangle\!\langle\Psi|)=\|\Psi\|_1^2-1$.

A complementary notion of delocalization for pure states, widely used in the exciton transport literature, is based only on the site-population distribution $w_n:=|\psi_n|^2$ and therefore quantifies how strongly the wavefunction weight is concentrated on a few sites (independent of relative phases). In contrast, the $\ell^1$-coherence $C_{\ell^1}(\ket{\Psi})$ \eqref{eq:l1-coherence} measures the total magnitude of site-basis coherences.
A common measure is the inverse participation ratio (IPR) of a normalized state $\ket{\Psi}=\sum_{n=1}^N \psi_n\ket{n}$, defined as \cite{Thouless1974,JDBF96,MZCM97,S20}
\begin{equation}
\mathrm{IPR}(\ket{\Psi}) := \sum_{n=1}^N |\psi_n|^4,
\end{equation}
and its inverse, the participation ratio $\mathrm{PR}(\ket{\Psi}) := 1/\mathrm{IPR}(\ket{\Psi})$.
Then, using the Frobenius norm $\|B(t)\|_{\mathrm{F}}:=\sqrt{\sum_{n,m=1}^N |B_{nm}(t)|^2}$ of $B(t)$, we obtain an analogue to Theorem \ref{thm:min-deloc} in terms of the measures IP and IPR.
\begin{cor}[Necessary IPR for a $\delta$-improvement]\label{cor:ipr_pr_CM}
Fix $t\geq 0$ and let $B(t)$ be defined as in Eq.~\eqref{eq:DB-decomp}. If a pure donor state $\ket{\Psi}$ achieves an improvement of $\delta\geq 0$ over the best incoherent value, i.e.,
\begin{equation}
\bra{\Psi}M_{\mathrm{trap}}(t)\ket{\Psi} \geq \eta_{\mathrm{trap}}^{(\mathrm{incoh})}(t) +\delta\,,
\end{equation}
with $\eta_{\mathrm{trap}}^{(\mathrm{incoh})}(t)$ given by \eqref{eq:eta-incoh-max-t}, then necessarily
\begin{equation}\label{eq:ipr_condition_CM}
\mathrm{IPR}(\ket{\Psi}) \leq  1-\frac{\delta^2}{\|B(t)\|_{\mathrm{F}}^2}\leq  1-\frac{\delta^2}{N\,\mathcal{C}_{M_{\mathrm{trap}}(t)}(\mathrm{id})^2}\,.
\end{equation}
\end{cor}
The proof of Corollary \ref{cor:ipr_pr_CM} is provided in Appendix \ref{app:cor-ipr_pr_CM}. 

Finally, we show that $\mathcal{C}_{M_{\mathrm{trap}}(t)}(\mathrm{id})$ also controls when the efficiency maximizing state must remain close to a localized site state:
\begin{thm}[Localization of the optimal trapping state]\label{thm:DK-chain}
Fix $t\geq 0$ and decompose $M_{\mathrm{trap}}(t)=D(t)+B(t)$ as in Eq.~\eqref{eq:DB-decomp}. Assume that the largest diagonal entry of $D(t)$ is attained uniquely at site $k$, and define the incoherent gap
\begin{equation}
\Delta_{\mathrm{incoh}}(t):= D_{kk}(t)-\max_{n\neq k}D_{nn}(t) > 0\,.
\end{equation}
Moreover, assume that $\mathcal{C}_{M_{\mathrm{trap}}(t)}(\mathrm{id})<\Delta_{\mathrm{incoh}}(t)$. Then any maximizer of $\eta_{\mathrm{trap}}^{(\max)}(t)$ can be chosen pure, $\rho_{\max}^{(\eta)}=\ket{\Psi_{\max}^{(\eta)}}\!\bra{\Psi_{\max}^{(\eta)}}$, and satisfies
\begin{equation}\label{eq:DK-overlap}
1-\big|\langle k|\Psi_{\max}^{(\eta)}\rangle\big|^2\leq    \left(\frac{\mathcal{C}_{M_{\mathrm{trap}}(t)}(\mathrm{id})}{\Delta_{\mathrm{incoh}}(t)-\mathcal{C}_{M_{\mathrm{trap}}(t)}(\mathrm{id})}\right)^2\,.
\end{equation}
\end{thm}
The proof of Theorem \ref{thm:DK-chain} is given in Appendix \ref{app:thm-DK-chain}. Theorem \ref{thm:DK-chain} provides a quantitative stability criterion for the structure of the optimal preparation for trapping at time $t$. 
For this, we recall that the diagonal part $D(t)=\mathcal{G}(M_{\mathrm{trap}}(t))$ identifies the best site-localized (incoherent) input via its largest entry $D_{kk}(t)$, while the off-diagonal part $B(t)$ captures the extent to which phase relations in the input can influence the trapped amount. If the incoherent gap $\Delta_{\mathrm{incoh}}(t)$ separating the best site $k$ from all others dominates the off-diagonal scale $\|B(t)\|_\infty=\mathcal{C}_{M_{\mathrm{trap}}(t)}(\mathrm{id})$, then any state maximizing $\eta_{\mathrm{trap}}^{(\max)}(t)$ must be close to the localized state $\ket{k}$ in the sense of Eq.~\eqref{eq:DK-overlap}. 
In particular, when $\mathcal{C}_{M_{\mathrm{trap}}(t)}(\mathrm{id})\ll \Delta_{\mathrm{incoh}}(t)$, the optimal state is forced to be predominantly site-localized and any potential improvement due to delocalization is perturbatively small. 

\begin{figure*}[ht]
\centering
\includegraphics[width=\linewidth]{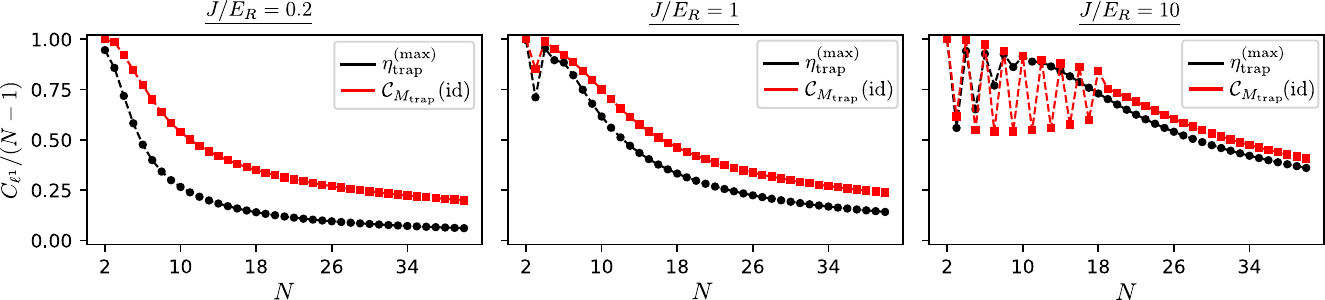}
\caption{Normalized $\ell^1$-coherence of $\ket{\Psi_{\max}^{(\eta)}}$ (black) and $\ket{\Psi_{\max}^{(\mathcal{C})}}$ (red) for the homogeneous one-dimensional chain model as a function of chain length $N$ for weak (left, $J/E_R=0.2$), intermediate (middle $J/E_R=1$), and strong (right, $J/E_R=10$) coupling. The systems parameters are $\varepsilon_n=0$, $J_n=J=100\,\mathrm{cm}^{-1}$, $T=300\,\mathrm{K}$, $\Gamma=0.01\,\mathrm{ps}^{-1}$, $\kappa=1\,\mathrm{ps}^{-1}$, and $\tau_c=0.01\,J^{-1}$. 
\label{fig:l1-coherence-eta-max-CM}}
\end{figure*}

\begin{figure*}[ht]
\centering
\includegraphics[width=\linewidth]{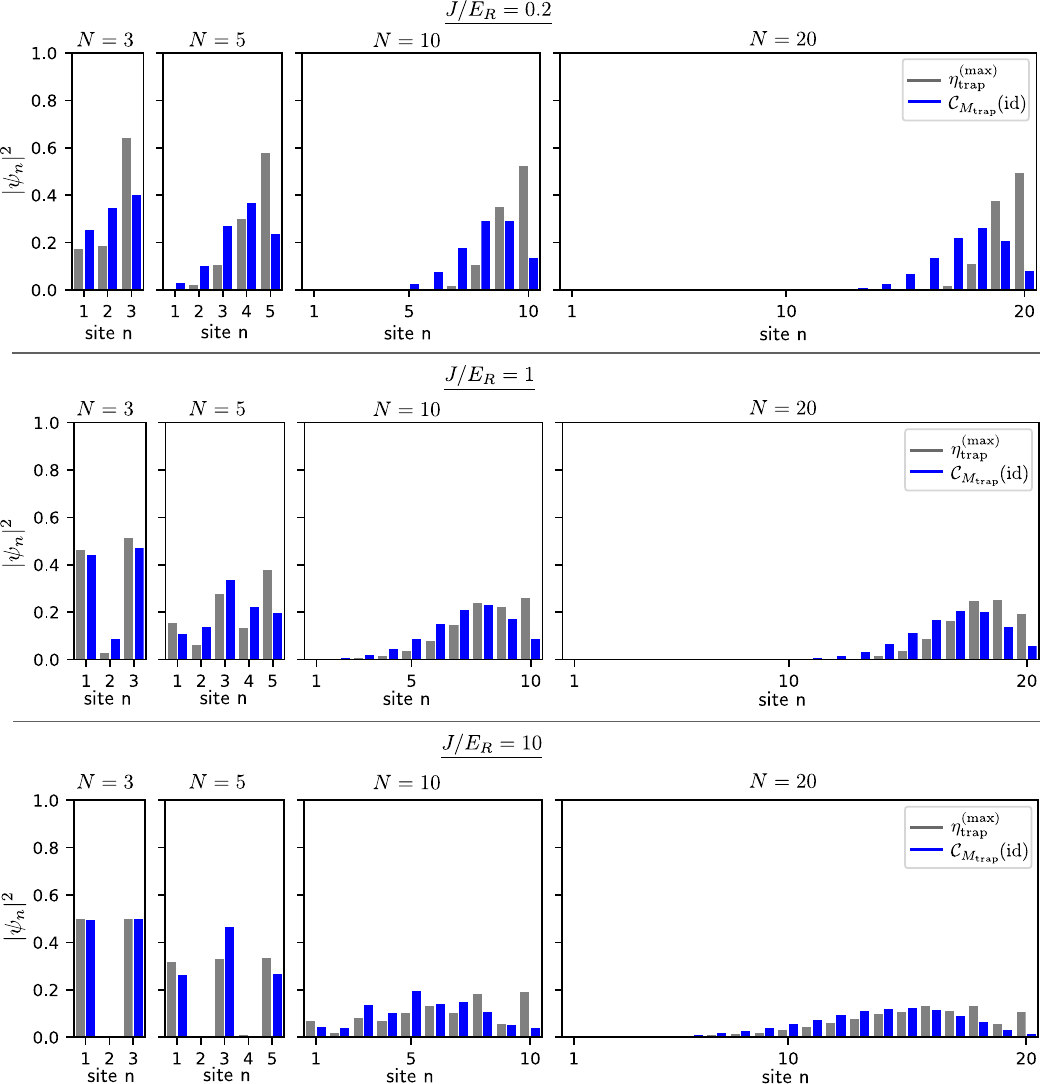}
\caption{Histogram of the weight distributions $|\psi_n|^2$ of the maximizers $\ket{\Psi_{\max}^{(\eta)}}$ of $\eta^{(\max)}_{\mathrm{trap}}$  (gray) and $\ket{\Psi_{\max}^{(\mathcal{C})}}$ of $\mathcal{C}_{M_{\mathrm{trap}}}(\mathrm{id})$ (blue) for the different coupling regimes $J/E_R\in \{0.2, 1, 10\}$ and chain lengths $N\in\{3,5,10,20\}$.
The fixed system parameters are the same as in Fig.~\ref{fig:l1-coherence-eta-max-CM}.
\label{fig:histogram_eta-max}}
\end{figure*}

\subsection{Optimal initial state analysis \label{sec:illustrations-chain}}

To assess when optimal trapping is driven by population placement versus site-basis coherence, we compare in the following the eigenstate that maximizes the trapped amount $\eta^{(\max)}_{\mathrm{trap}}$ \eqref{eq:eta-max-chain} with the eigenstate that maximizes the functional $\mathcal{C}_{M_{\mathrm{trap}}}(\mathrm{id}) = \mathcal{C}_{|s\rangle\!\langle s|}(\Lambda_t)$.

For this, we compare in Figs.~\ref{fig:l1-coherence-eta-max-CM} and~\ref{fig:histogram_eta-max} two distinguished pure donor states associated with the trapping functional at $t\to\infty$:
(i) the maximizer of the trapped amount, $\rho_{\max}^{(\eta)}=\ket{\Psi_{\max}^{(\eta)}}\!\bra{\Psi_{\max}^{(\eta)}}$, which attains $\eta_{\rm trap}^{(\max)}=\|M_{\rm trap}\|_\infty$, and (ii) a state $\rho_{\max}^{(\mathcal{C})}=\ket{\Psi_{\max}^{(\mathcal{C})}}\!\bra{\Psi_{\max}^{(\mathcal{C})}}$ that maximizes the coherence sensitivity of the same task, i.e., attains $\mathcal C_{M_{\rm trap}}(\mathrm{id})=\|(\mathrm{id}-\mathcal G)(M_{\rm trap})\|_\infty$.
Since $\eta_{\rm trap}^{(\max)}$ depends on both the diagonal and off-diagonal parts of $M_{\rm trap}$, whereas $\mathcal C_{M_{\rm trap}}(\mathrm{id})$ depends only on its off-diagonal part, these two optimizers need not coincide.

We quantify the delocalization of each optimizer by the normalized $\ell^1$-coherence ratio $C_{\ell^1}(\ket{\Psi}\!\bra{\Psi})/(N-1)$ (recall Eq.~\eqref{eq:l1-coherence}).
Figure \ref{fig:l1-coherence-eta-max-CM} shows this ratio for $\ket{\Psi_{\max}^{(\eta)}}$ (black) and $\ket{\Psi_{\max}^{(\mathcal{C})}}$ (red) as a function of chain length $N$, for the weak ($J/E_R=0.2$), intermediate ($J/E_R=1$), and strong ($J/E_R=10$) coupling regimes.
Across all regimes, the optimizer of $\eta_{\rm trap}^{(\max)}$ becomes progressively less delocalized as $N$ increases, reflecting the fact that maximizing the final trapped amount favors placing weight closer to the terminal trapping site when recombination and dephasing are present.
In contrast, the optimizer associated with $\mathcal C_{M_{\rm trap}}(\mathrm{id})$ remains on average slightly more delocalized, indicating that the trapping task is, in principle, most sensitive to coherent superpositions that involve multiple sites even when the globally optimal trapped amount can be achieved with a more localized preparation.

To connect these trends to the spatial structure of the optimizers, Fig.~\ref{fig:histogram_eta-max} plots the site populations $|\psi_n|^2$ of $\ket{\Psi_{\max}^{(\eta)}}$ (gray) and $\ket{\Psi_{\max}^{(\mathcal{C})}}$ (blue) for selected $N\in \{3, 5, 10, 50\}$.
For small, disorder-free chains (for instance, $N=3$) both optimizers exhibit a clear parity structure, with suppressed weight on even numbered sites for $N$ odd and small enough, which explains the corresponding dips in the coherence ratio as a function of $N$ as shown in Fig.~\ref{fig:l1-coherence-eta-max-CM}.
For larger $N$, $\ket{\Psi_{\max}^{(\eta)}}$ is increasingly biased towards the end of the chain, whereas $\ket{\Psi_{\max}^{(\mathcal{C})}}$ typically retains a broader support on a larger number of sites.
This difference is expected because $\eta_{\mathrm{trap}}^{(\max)}$ can be increased by redistributing populations (the diagonal part of $M_{\mathrm{trap}}$), while $\mathcal{C}_{M_{\mathrm{trap}}}(\mathrm{id})$ isolates how strongly phase relations at fixed populations can modulate the trapped amount (the off-diagonal part).

Although the two optimizers probe different aspects of the trapping problem, comparing their structure is still informative. The quantity $\mathcal{C}_{M_{\mathrm{trap}}}(\mathrm{id})$ provides a state-independent upper bound on the magnitude of any coherence-induced change in the trapped amount as well as any coherence-induced improvement over the best incoherent preparation (recall Corollary \ref{thm:eta-max-coh}). Fig.~\ref{fig:l1-coherence-eta-max-CM} contrasts the site-basis delocalization of the state that maximizes the trapping objective ($\ket{\Psi_{\max}^{(\eta)}}$) with that of the state that maximizes coherence sensitivity ($\ket{\Psi_{\max}^{(\mathcal{C})}}$). If the two curves are close, the trapping optimum is achieved by a state whose delocalization is comparable to that of the coherence-sensitive optimal state, suggesting that coherence may be operationally relevant already at the level of optimal state design. Conversely, if $\ket{\Psi_{\max}^{(\eta)}}$ is substantially more localized than $\ket{\Psi_{\max}^{(\mathcal{C})}}$, then near-optimal trapping can be obtained primarily through population placement (diagonal contributions), while coherence remains an available but unexploited degree of freedom: the dynamics is in principle sensitive to appropriately structured superpositions (as indicated by $\mathcal{C}_{M_{\mathrm{trap}}}(\mathrm{id})>0$), but exploiting this sensitivity requires preparations that are specifically aligned with the off-diagonal pattern encoded by $(\mathrm{id}-\mathcal{G})(M_{\mathrm{trap}})$.

A natural next step would be to investigate how the optimizing states $\ket{\Psi_{\max}^{(\mathcal{C})}}$ and $\ket{\Psi_{\max}^{(\eta)}}$, and in particular their weight distributions, behave under static disorder. Moreover, beyond single realizations, an ensemble perspective would enable a probabilistic characterization of the optimizers.

\subsection{A coherence light cone from quasi-locality\label{sec:coh-lightcone}}

In extended photosynthetic antennae and other excitonic aggregates, photoexcitation can be created far from the trap, such that transport and trapping involve propagation across many sites  \cite{HVUPJS93, CS15, MOAGGS17, CCCDHKJM20, TNKKKS25}. 
Whether site-basis coherence in the initially excited domain can influence a localized trapping readout is therefore constrained not only by decoherence times, but also by the finite speed at which any phase-sensitive influence can propagate through a locally coupled network. 
Initial coherence may affect the readout either directly, through phase-sensitive amplitudes reaching the sink region, or indirectly by imprinting on populations that subsequently migrate. However, both mechanisms are ultimately limited by quasi-locality. 
Such propagation constraints can be formalized using Lieb-Robinson bounds \cite{LR72, Poulin10, BK12}, which quantify an effective ``light cone'' for the spread of operator support and, hence, of readout sensitivity, under local dynamics.

Let $V$ be the set of sites (optionally including a trap/sink site $s$) with site basis $\{\ket{v}\}_{v\in V}$. For $S\subseteq V$ define $P_S:=\sum_{v\in S}\ket{v}\!\bra{v}$ and
\begin{equation}
\mathcal D(S):=\{\rho\in\mathcal D(\mathcal H)\mid \rho=P_S\rho P_S\}\,.
\end{equation}
Let $\mathcal G$ denote complete dephasing in the site basis, $\mathcal G(\rho)=\sum_{v\in V}\ket{v}\!\bra{v}\rho\ket{v}\!\bra{v}$.
For a Hermitian readout observable $M$ and a (not necessarily semigroup) dynamical map $\Lambda_t$, we define the support-restricted impact functional
\begin{equation}
\mathcal C_M^{(S)}(\Lambda_t) :=\sup_{\rho\in\mathcal D(S)} \left|\Tr \left[M\,\Lambda_t\!\left(\rho-\mathcal G(\rho)\right)\right]\right|\,.
\end{equation}
In analogy to Eq.~\eqref{eq:capacity}, $\mathcal C_M^{(S)}(\Lambda_t)$ can be written as the operator norm of the corresponding block, that is
\begin{equation}\label{eq:C_restricted_opnorm}
\mathcal C_M^{(S)}(\Lambda_t) =\left\|P_S(\mathrm{id}-\mathcal G)\bigl(\Lambda_t^\dagger(M)\bigr)P_S\right\|_\infty\,.
\end{equation}

We equip $V$ with a graph structure, where edges represent non-negligible effective couplings, and write $d(X,Y)$ for the induced graph distance between subsets $X,Y\subseteq V$. Moreover, we assume that the reduced Heisenberg evolution $\Lambda_t^\dagger$ is quasi-local in the following sense: there exist constants $C_{\mathrm{QL}}\geq 0$, $v_{\mathrm{QL}}\geq 0$, and $\mu > 0$ such that for every subset $X \subseteq V$, every operator $O$ supported on $X$ (i.e., $O=P_X O P_X$), and every integer $r\geq 0$, there exists an operator $O^{(r)}(t)$ supported on the $r$-neighbourhood $X[r]:=\{v\in V\,|\,d(v,X)\leq r\}$ with
\begin{equation}\label{eq:quasilocality_assumption}
\left\|\Lambda_t^\dagger(O)-O^{(r)}(t)\right\|_\infty \leq C_{\mathrm{QL}}\,\|O\|_\infty\,\mathrm{e}^{v_{\mathrm{QL}} t-\mu r}\,. 
\end{equation}
Estimates of this form follow from Lieb-Robinson bounds for local Markovian (GKLS) dynamics \cite{Poulin10, BK12}. Beyond the Markovian setting, related bounds have been obtained for certain classes of finite-memory non-Markovian environments \cite{TYR24}. For long-range interactions, generalizations typically involve modified (algebraic) spatial decays, while here we work with the exponential form in Eq.~\eqref{eq:quasilocality_assumption} appropriate for finite-range or sufficiently fast decaying effective interactions \cite{SEK19}. 

\begin{thm}[Coherence light cone for local readouts]\label{thm:lightcone}
Assume \eqref{eq:quasilocality_assumption} holds. Let $M=M^\dagger$ be supported on $X\subseteq V$ and let
$S\subseteq V$ satisfy $d(S,X)\geq 1$. Then for all $t\geq 0$,
\begin{equation}\label{eq:lightcone_bound}
\mathcal{C}_M^{(S)}(\Lambda_t) \leq C_{\mathrm{LC}}\,\|M\|_\infty \exp\left(v_{\mathrm{QL}} t-\mu\, d(S,X)\right)\,,
\end{equation}
with $C_{\mathrm{LC}}:=2C_{\mathrm{QL}}e^{\mu}$.
In particular, for fixed $t$ the upper bound decays exponentially in the distance $d(S,X)-v_{\mathrm{QL}}t/\mu$ whenever $d(S,X) > v_{\mathrm{QL}}t/\mu$.
\end{thm}

\begin{cor}\label{cor:coh_lightcone_chain}
Consider a one-dimensional chain $V=\{1,\dots,N\}\cup\{s\}$ where the sink $s$ is adjacent to site $N$,
and let the readout be $M=\ket{s}\!\bra{s}$ which is supported on $X=\{s\}$.
For an initial domain $S\subseteq\{1,\dots,N\}$ define $d_S:=d(S,\{N\})$. Then $d(S,\{s\})=d_S+1$, and
Theorem~\ref{thm:lightcone} yields, for all $t\geq 0$,
\begin{equation}
\mathcal{C}_{|s\rangle\!\langle s|}^{(S)}(\Lambda_t)\leq C_{\mathrm{LC}}\,\exp\!\big(v_{\mathrm{QL}} t-\mu(d_S+1)\big)\,.
\end{equation}
Equivalently, with $v_{\mathrm{LC}}:=v_{\mathrm{QL}}/\mu$, the influence of any initial site-basis coherence
supported on $S$ on the sink population is exponentially suppressed in $(d_S+1)-v_{\mathrm{LC}}t$.
In particular, for $t<(d_S+1)/v_{\mathrm{LC}}$ this influence is exponentially small in the distance to the sink.
\end{cor}
The proof of Theorem \ref{thm:lightcone} is shown in Appendix \ref{app:thm-lightcone}, and Corollary \ref{cor:coh_lightcone_chain} follows as an application to the chain geometry. 
In pigment protein complexes, functional readouts such as reaction center trapping are typically localized to a small subnetwork, while excitation can prepare states supported on other domains. Theorem~\ref{thm:lightcone} provides a state-independent sufficient condition for negligible coherence influence: if $d(S,X)$ is large compared to $v_{\mathrm{LC}}t$ with the light cone velocity $v_{\mathrm{LC}}:=v_{\mathrm{QL}}/\mu$, then the upper bound on the maximal contribution of site-basis coherence initially confined to $S$ to any readout supported on $X$ is exponentially small in $d(S,X)-v_{\mathrm{LC}}t$. Moreover, the difference between the optimal coherent and incoherent preparation on $S$ at time $t$ is bounded from above by $\mathcal{C}_M^{(S)}(\Lambda_t)$.

This locality-based viewpoint complements the familiar coherence-lifetime argument in excitation energy transfer, which focuses on how rapidly off-diagonal density-matrix elements decay. Lieb-Robinson-type bounds do not quantify how rapidly coherences decay. Instead they constrain how quickly a localized readout can become sensitive to phases initially prepared in a distant domain. 
These bounds are most informative in extended networks, where the distance $d(S,X)$ between the preparation domain $S$ and the readout region $X$ can be large.
Theorem~\ref{thm:lightcone} shows that the maximal coherence-induced change in the readout is exponentially suppressed while $X$ remains outside the light cone of $S$, i.e., whenever $d(S,X)>v_{\mathrm{LC}}t$. For localized preparation and trapping, the theorem should be viewed as a state-independent causality bound on the operational relevance of remote initial coherence. It complements, rather than replaces, model-specific estimates of excitation transfer to the trapping region.
To relate this to the usual timescale picture, it is useful to introduce, in addition to the coherence time $\tau_{\mathrm{coh}}$, a memory time $\tau_{\mathrm{mem}}$ for the coherence-induced perturbation.
Consider the two initial states $\rho_0$ and $\mathcal{G}(\rho_0)$.
We define $\tau_{\mathrm{mem}}$ as the timescale on which $\Lambda_t(\rho_0)$ and $\Lambda_t(\mathcal{G}(\rho_0))$ become indistinguishable within the source region $S$, including any coherence-induced redistribution of populations.
A convenient choice is the decay time of the trace distance between the reduced states on $S$.
If $\tau_{\mathrm{mem}}\ll d(S,X)/v_{\mathrm{LC}}$, up to the logarithmic correction set by the prefactor, then the coherence-induced perturbation is washed out before the readout on $X$ can become sensitive to it and its contribution to localized reaction-center observables is negligible.

Finally, Eq.~\eqref{eq:lightcone_bound} provides a practical truncation criterion.
To compute $\mathcal C_M^{(S)}(\Lambda_t)$ at time $t$ within accuracy $\varepsilon$, it suffices to simulate the reduced dynamics on a neighborhood of $X$ of radius $r$ chosen such that $C_{\mathrm{QL}}\|M\|_\infty e^{v_{\mathrm{QL}}t-\mu r}\lesssim \varepsilon$, i.e., $r$ needs to grow at most linearly with $t$, in direct analogy with standard Lieb-Robinson light cone estimates.

\section{Conclusions and outlook}

We have developed a complementary, process level perspective on coherence in excitation energy transfer based on quantum information tools and the resource impact functional $\mathcal{C}_M(\Lambda_t)$ introduced in Ref.~\cite{LS26-RT}. This functional provides a readout specific and state independent bound on the largest change that initial site-basis coherence can induce in the expectation value of a chosen observable $M$ under a fixed reduced dynamics $\Lambda_t$. In this sense, $\mathcal{C}_M(\Lambda_t)$ turns the qualitative question of ``coherence relevance'' into a quantitative, task-dependent estimate that can be evaluated for both Markovian and non-Markovian models. 
In particular, resource theory enters here through the construction of the resource-free reference state $\mathcal G(\rho)$, rather than through a restriction of the physical dynamics to free operations.

In the donor-acceptor setting, we have evaluated $\mathcal{C}_M(\Lambda_t)$ across coupling and bath timescale regimes, including intermediate coupling where perturbative master equations can be quantitatively unreliable. The resulting time-resolved impact functional recovers established qualitative trends: it identifies early time windows of enhanced phase sensitivity and shows how dephasing and loss suppress coherence sensitivity at later times. Importantly, these statements are
state-independent. They therefore complement earlier state-specific analyses and directly address concerns about preparation dependence raised in the spectroscopy and transport literature.

For the $N$-site donor chain with trapping, we derived several rigorous criteria that separate population placement effects from genuine coherence sensitivity. Theorem \ref{thm:eta-max-coh} bounds the largest possible advantage of coherent over incoherent initial preparations for trapping in terms of $\mathcal{C}_{M_{\mathrm{trap}}(t)}(\mathrm{id})$, yielding an operational no-go criterion whenever this quantity is small. 
Corollary~\ref{cor:2site_lb} provides a constructive and easily computable sufficient condition for coherent advantage based on interference between two-site superpositions, even though the globally optimal state may involve many sites.
In addition, Theorem \ref{thm:min-deloc} and Corollary \ref{cor:ipr_pr_CM} quantify the minimal amount of delocalization required to achieve a prescribed improvement over incoherent preparations, expressed via $\ell^1$-coherence and participation measures. 
Finally, Theorem~\ref{thm:DK-chain} gives a stability criterion: when an incoherent strategy is separated by a sufficiently large diagonal gap, the trapping maximizer is forced to remain close to a localized site state. Together with the numerical optimizer analysis, these results clarify how the optimal preparation shifts with chain length and dissipation, and when coherence sensitivity is present but not exploitable without fine-tuned superpositions.

Several extensions are natural and, for realistic molecular aggregates, essential. The homogeneous chains considered here represent an idealized limit. Static disorder and spatially varying couplings are ubiquitous in pigment protein complexes and in synthetic aggregates, and they can change both the optimal preparation and the magnitude of coherent advantages. Extending the present analysis to disordered chains, for example by av eraging over disorder realizations in a Monte Carlo approach and deriving the corresponding ensemble-averaged bounds, is therefore an important next step. On the methodological side, a key advantage of $\mathcal{C}_M(\Lambda_t)$ is that many bounds derived here only assume linear reduced dynamics and a fixed resource-destroying map. This makes them portable across simulation methods, including HEOM \cite{TK89, IF09-JCP, Tanimura20, IS20}, polaron based approaches \cite{SH84, JCRE08, Jang11, MCN11, PMCLGN13, TJC19, Jang22, XC16}, and other numerically stable reduced descriptions.

A further outcome of this work is the connection between coherence sensitivity and quasi-locality.
The ``coherence light cone'' of Theorem \ref{thm:lightcone} provides a finite-time locality constraint: it does not quantify coherence decay in the state, but it bounds how quickly phase sensitivity of a localized readout can be generated from coherence initially confined to a distant domain. To make such bounds actionable in specific complexes, one needs estimates of the effective parameters entering the quasi-locality inequality (e.g., velocities and decay rates). Even if these bounds are conservative, they provide a systematic way to combine geometry with timescale arguments often invoked in the excitation energy transfer literature. They also motivate controlled truncations: for fixed accuracy, only a neighbourhood of the readout region needs to be simulated, with a radius that grows at most linearly in time under quasi-local dynamics.

The present resource impact functional has been evaluated here for specified dynamical models.
Experimentally, state-specific values of $\Delta Y_M(\rho)$ could be obtained by comparing paired preparations $\rho$ and $\mathcal G(\rho)$, when such preparations are available, and would provide lower bounds on $\mathcal C_M(\Lambda_t)$. Determining the full supremum requires more information, for example a reconstruction of the reduced map or of the relevant linear input-output functional from sufficiently many independent preparations and measurements.
Many spectroscopic probes of molecular coherence are multi-time measurements rather than single-time readouts.
Applying the present idea directly to such experiments would therefore require a multi-time extension with explicit intermediate interventions. A natural framework for this is the process tensor \cite{PRFPM18, PMRFP18, dWKLC25, KSBR25}, or a restricted process tensor when only a limited set of experimental controls is available.
In molecular spectroscopy, this would likely lead to protocol-dependent versions of the resource impact functional rather than a fully model-independent reconstruction.
Such extensions could provide a way to compare HEOM or other model predictions with experimentally reconstructed or partially reconstructed multi-time dynamics, and remain an important direction for future work.

Finally, the generality of the resource impact functional suggests several routes to connect with current developments on vibrational and vibronic effects in energy transfer \cite{MSNMvG16, AYINF20, HLSAOS21, KNPGCSC21, ZRLZWCWW22}. In many systems, long-lived oscillatory signals arise from vibronic mixing or vibrational coherences rather than purely electronic coherences \cite{CKPM12, TPJ13, WM11, LPCSLP15, JRWS18}. Within the present framework, this can be addressed by enlarging the system Hilbert space to include selected vibrational modes and by choosing resource-destroying maps that target different notions of coherence such as electronic coherence in a site basis, vibrational coherence in a mode basis, or coherence between vibronic eigenstates. This opens the door to a unified, readout-specific comparison of electronic, vibrational, and vibronic contributions to trapping or charge separation observables, and to assessing when such coherences can be operationally relevant within realistic finite time windows and preparation constraints.

\begin{acknowledgments}
This research was funded by the National Science Foundation under Grant No.~CHE-2537080.
\end{acknowledgments}

\appendix

\section{Proof of Theorem \ref{thm:eta-max-coh} \label{app:thm-eta-max-coh}}

\begin{proof}
We first decompose $M_{\mathrm{trap}}(t)$ into its diagonal and off-diagonal parts in the site basis,
\begin{equation}
D(t):=\mathcal{G}\!\big(M_{\mathrm{trap}}(t)\big)\,,\quad
B(t):=(\mathrm{id}-\mathcal{G})\!\big(M_{\mathrm{trap}}(t)\big)
\end{equation}
such that $M_{\mathrm{trap}}(t)=D(t)+B(t)$.
Since $D(t)$ is diagonal, the optimization over incoherent states reduces to an optimization over probability vectors on the diagonal, hence
\begin{equation}\label{eq:eta-incoh-max-diag}
\eta_{\mathrm{trap}}^{(\mathrm{incoh})}(t)
=\max_{n\in\{1,\dots,N\}} D_{nn}(t)
=\lambda_{\max}\!\big(D(t)\big),
\end{equation}
where $D_{nn}(t)=\bra{n}D(t)\ket{n}$ and $\lambda_{\max}$ denotes the largest eigenvalue.
Moreover, by the Rayleigh-Ritz variational principle,
\begin{equation}\label{eq:eta-max-eig}
\eta_{\mathrm{trap}}^{(\max)}(t) =\max_{\rho\in\mathcal{D}(\mathcal{H})}\Tr\big[M_{\mathrm{trap}}(t)\rho\big]
=\lambda_{\max} \big(M_{\mathrm{trap}}(t)\big)\,.
\end{equation}
Then, Weyl's inequality for Hermitian matrices (see, e.g., Ref.~\cite{Bhatia}) yields
\begin{equation}
\lambda_{\max}\!\big(D(t)+B(t)\big)\leq \lambda_{\max}\!\big(D(t)\big)+\|B(t)\|_\infty\,,
\end{equation}
where $\|\cdot\|_\infty$ denotes the operator norm. Combining with Eqs.~\eqref{eq:eta-incoh-max-diag} and~\eqref{eq:eta-max-eig} gives
\begin{equation}
\eta_{\mathrm{trap}}^{(\max)}(t)-\eta_{\mathrm{trap}}^{(\mathrm{incoh})}(t) \leq \|B(t)\|_\infty = \mathcal{C}_{M_{\mathrm{trap}}(t)}(\mathrm{id})\,,
\end{equation}
which proves Eq.~\eqref{eq:eta-gap-bound}.
\end{proof}

\section{Proof of Corollary \ref{cor:2site_lb} \label{app:cor-2site_lb}}
\begin{proof}
Let $P_{nm}:=\ket{n}\!\bra{n}+\ket{m}\!\bra{m}$ be the projector onto $\mathrm{span}\{\ket{n},\ket{m}\}$. Then, 
\begin{align}
\eta_{\mathrm{trap}}^{(\max)}(t) &=\max_{|\Psi\rangle}\bra{\Psi}M_{\mathrm{trap}}(t)\ket{\Psi}\nonumber\\
&\geq\max_{\substack{|\Psi\rangle\in\mathrm{Ran}(P_{nm})\\ \|\Psi\|=1}}\bra{\Psi}M_{\mathrm{trap}}(t)\ket{\Psi}\nonumber\\
&= \lambda_{\max}\!\big(P_{nm}M_{\mathrm{trap}}(t)P_{nm}\big)\,,
\end{align}
where $\lambda_{\max}$ denotes the maximal eigenvalue.
In the basis $\{\ket{n},\ket{m}\}$, the matrix of $P_{nm}M_{\mathrm{trap}}(t)P_{nm}$ is
\begin{equation}
\begin{pmatrix}
D_{nn}(t) & B_{nm}(t)\\
B_{nm}(t)^* & D_{mm}(t)
\end{pmatrix},
\end{equation}
whose largest eigenvalue is the closed form expression $\lambda_{nm}^{(+)}(t)$, proving
\eqref{eq:2site_lb_main}.
\end{proof}

\section{Proof of Theorem~\ref{thm:min-deloc}\label{app:thm-min-deloc}}

\begin{proof}
We fix $t\ge 0$ and abbreviate $M(t):=M_{\mathrm{trap}}(t)$, which is decomposed as $M(t)=D(t)+B(t)$ with
$D(t):=\mathcal G(M(t))$ and $B(t):=(\mathrm{id}-\mathcal G)(M(t))$.
Let $d_{\max}(t):=\eta_{\mathrm{trap}}^{(\mathrm{incoh})}(t)=\max_n D_{nn}(t)$ and assume that a given normalized pure state $\ket{\Psi} = \sum_n\psi_n\ket{n}$ satisfies
\begin{equation}\label{eq:app-assump}
\bra{\Psi}M(t)\ket{\Psi} \geq d_{\max}(t)+\delta
\end{equation}
for some $\delta\geq 0$. Since $D(t)$ is diagonal and $\sum_n |\psi_n|^2=1$, we have
\begin{equation}
\bra{\Psi}D(t)\ket{\Psi}=\sum_{n} D_{nn}(t)\,|\psi_n|^2 \leq \max_n D_{nn}(t)=d_{\max}(t)\,.
\end{equation}
Subtracting this inequality from \eqref{eq:app-assump} yields
\begin{equation}\label{eq:app-B-lb}
\bra{\Psi}B(t)\ket{\Psi} \geq \delta\,.
\end{equation}
In particular, this implies that $\delta \leq |\bra{\Psi}B(t)\ket{\Psi}|$. Using the definition of $c_{\max}(t)$,
\begin{equation}
c_{\max}(t):=\max_{n\neq m}|B_{nm}(t)|\,,
\end{equation}
we further obtain
\begin{align}
|\bra{\Psi}B(t)\ket{\Psi}| &= \left|\sum_{n\neq m} B_{nm}(t)\,\psi_n^* \psi_m\right|\nonumber\\
&\leq \sum_{n\neq m} |B_{nm}(t)|\,|\psi_n|\,|\psi_m| \nonumber\\
&\leq  c_{\max}(t)\,C_{\ell^1}(\ket{\Psi}\!\bra{\Psi})\,,
\end{align}
where in the last step we used $C_{\ell^1}(\ket{\Psi}\!\bra{\Psi})=\sum_{n\neq m}|\psi_n\psi_m|=\|\Psi\|_1^2-1$.
Combining this with \eqref{eq:app-B-lb} yields
\begin{equation}
C_{\ell^1}(\ket{\Psi}\!\bra{\Psi}) \geq \frac{\delta}{c_{\max}(t)}\,.
\end{equation}
Since $c_{\max}(t)\leq   \|B(t)\|_\infty=\mathcal C_{M_{\mathrm{trap}}(t)}(\mathrm{id})$, it follows that
\begin{equation}\label{eq:l1-app-bound}
C_{\ell^1}(\ket{\Psi}\!\bra{\Psi})\geq \frac{\delta}{c_{\max}(t)} \geq \frac{\delta}{\mathcal C_{M_{\mathrm{trap}}(t)}(\mathrm{id})}\,.
\end{equation}
Finally, if $\ket{\Psi}$ is supported on $m$ sites, then by Cauchy-Schwarz $\|\Psi\|_1 \leq \sqrt{m}\,\|\Psi\|_2=\sqrt{m}$, and hence $\|\Psi\|_1^2\leq m$. This implies that $m$ is lower bounded in terms of $c_\mathrm{max}(t)$ and $\delta$ according to
\begin{equation}
m \geq \|\Psi\|_1^2 = 1 + C_{\ell^1}(\ket{\Psi}\!\bra{\Psi})\geq 1+\frac{\delta}{c_{\max}(t)}\,.
\end{equation}
\end{proof}

\section{Proof of Theorem~\ref{thm:DK-chain}\label{app:thm-DK-chain}}

\begin{proof}
Fix $t\ge 0$ and write $M(t):=M_{\mathrm{trap}}(t)=D(t)+B(t)$ as in Eq.~\eqref{eq:DB-decomp}, where $D(t)$ is diagonal in the site basis and $B(t)$ has vanishing diagonal. Let
\begin{equation}
c(t):=\|B(t)\|_\infty=\mathcal{C}_{M_{\mathrm{trap}}(t)}(\mathrm{id})
\end{equation}
and 
\begin{equation}\label{eq:gap-assumption}
\Delta_{\mathrm{incoh}}(t):=D_{kk}(t)-\max_{n\neq k}D_{nn}(t)>0\,,
\end{equation}
and assume that $c(t)<\Delta_{\mathrm{incoh}}(t)$.

Since $\rho\mapsto \Tr[M(t)\rho]$ is linear on the convex set $\mathcal{D}_D$, a maximizer of $\eta_{\mathrm{trap}}^{(\max)}(t)$ can be chosen pure. Let $\ket{\Psi_{\max}^{(\eta)}}$ be a normalized eigenvector of $M(t)$ corresponding to the largest eigenvalue 
\begin{align}
\lambda:=\lambda_{\max}(M(t))&=\eta_{\mathrm{trap}}^{(\max)}(t)\,,\,\,\\
M(t)\ket{\Psi_{\max}^{(\eta)}}&=\lambda\,\ket{\Psi_{\max}^{(\eta)}}\,.\label{eq:eigsystem-Mtrap}
\end{align}
Moreover, Weyl's inequality for Hermitian matrices \cite{Bhatia} implies
\begin{equation}
|\lambda_{\max}(M(t))-\lambda_{\max}(D(t))|\le \|B(t)\|_\infty=c(t).
\end{equation}
Since $\lambda_{\max}(D(t))=D_{kk}(t)$, this yields
\begin{equation}\label{eq:DK-proof-lb}
\lambda \geq D_{kk}(t) - c(t) \,.
\end{equation}
For any $n\neq k$, the gap assumption \eqref{eq:gap-assumption} implies that $D_{nn}(t)\leq D_{kk}(t)-\Delta_{\mathrm{incoh}}(t)$, hence combining with \eqref{eq:DK-proof-lb} gives
\begin{align}\label{eq:DK-proof-sep}
\lambda - D_{nn}(t)\geq   \Delta_{\mathrm{incoh}}(t) - c(t)   \quad \forall\,n\neq k\,.
\end{align}

Next, we decompose
\begin{equation}
\ket{\Psi_{\max}^{(\eta)}}=\alpha\,\ket{k}+\ket{x}\,,\quad \langle k|x\rangle =  0\,, \quad |\alpha|^2+\|x\|^2 =   1  \,,
\end{equation}
such that
\begin{equation}\label{eq:DK-proof-orth}
1 - \big|\langle k|\Psi_{\max}^{(\eta)}\rangle\big|^2 =1-|\alpha|^2   = \|x\|^2\,.
\end{equation}
Moreover, let $Q:=\mathbbm{1}-\ket{k}\!\bra{k}$ be the orthogonal projector onto $\ket{k}^\perp$. Applying $Q$ to the eigenvalue equation in Eq.~\eqref{eq:eigsystem-Mtrap} yields
\begin{equation}
Q[D(t)+B(t)]\ket{\Psi_{\max}^{(\eta)}} = \lambda Q\ket{\Psi_{\max}^{(\eta)}}\,. 
\end{equation}
Since $QD(t)\ket{k}=0$ and $Q\ket{\Psi_{\max}^{(\eta)}}=\ket{x}$, we obtain
\begin{equation}\label{eq:DK-proof-proj}
(\lambda\mathbbm{1}-D^\perp(t))\ket{x}=Q B(t)\ket{\Psi_{\max}^{(\eta)}}\,,
\end{equation}
where $D^\perp(t):=Q D(t) Q$ is the restriction of $D(t)$ to $\ket{k}^\perp$.

Because $D(t)$ is diagonal in the site basis, the eigenvalues of $D^\perp(t)$ are $\{D_{nn}(t):n\neq k\}$. Hence \eqref{eq:DK-proof-sep} implies that $\lambda\mathbbm{1}-D^\perp(t)$ is invertible on $\ket{k}^\perp$ and
\begin{align}\label{eq:DK-proof-inv}
\|(\lambda\mathbbm{1}-D^\perp(t))^{-1}\|_\infty &=  \frac{1}{\min_{n\neq k}(\lambda-D_{nn}(t))} \nonumber\\
&\leq \frac{1}{\Delta_{\mathrm{incoh}}(t)-c(t)}\,.
\end{align}
Taking norms in \eqref{eq:DK-proof-proj}, using that $Q$ is a contraction, and $\|\ket{\Psi_{\max}^{(\eta)}}\|=1$, we obtain
\begin{align}
\|x\| &\leq \|(\lambda\mathbbm{1}-D^\perp(t))^{-1}\|_\infty  \|Q B(t)\ket{\Psi_{\max}^{(\eta)}}\| \nonumber\\
&\leq \frac{c(t)}{\Delta_{\mathrm{incoh}}(t)-c(t)}\,.
\end{align}
Then, it follows from \eqref{eq:DK-proof-orth} that
\begin{align}
1-\big|\langle k|\Psi_{\max}^{(\eta)}\rangle\big|^2\leq \left(\frac{\mathcal{C}_{M_{\mathrm{trap}}(t)}(\mathrm{id})}{\Delta_{\mathrm{incoh}}(t)-\mathcal{C}_{M_{\mathrm{trap}}(t)}(\mathrm{id})}\right)^2\,,
\end{align}
which is Eq.~\eqref{eq:DK-overlap}.
\end{proof}

\section{Proof of Corollary \ref{cor:ipr_pr_CM} \label{app:cor-ipr_pr_CM}}

\begin{proof}
Let $\ket{\Psi}=\sum_{n=1}^N \psi_n \ket{n}$ be a normalized donor state and set $\rho=\ket{\Psi}\!\bra{\Psi}$ which satisfies the assumptions in Corollary \ref{cor:ipr_pr_CM}. Then, 
\begin{equation}
\langle\Psi|B(t)|\Psi\rangle  = \langle\Psi|M_{\mathrm{trap}}(t)|\Psi\rangle - \langle\Psi|D(t)|\Psi\rangle\geq \delta\,, 
\end{equation}
and, thus, $|\Tr[B(t)\rho]|\geq \delta$.
Since in addition $B(t)$ has vanishing diagonal, it follows that 
\begin{equation}
|\Tr[B(t)\rho]| \leq \|B(t)\|_{\mathrm{F}} \|(\mathrm{id}-\mathcal{G})(\rho)\|_{\mathrm{F}}\,,
\end{equation}
where $\|\cdot\|_{\mathrm{F}}$ denotes the Frobenius norm. Moreover, in the site basis $\rho_{nm}=\psi_n\psi_m^*$ such that
\begin{align}
\|(\mathrm{id}-\mathcal{G})(\rho)\|_{\mathrm{F}}^2 & =\sum_{n\neq m}|\psi_n|^2|\psi_m|^2  =1-\mathrm{IPR}(\ket{\Psi})\,,
\end{align}
which leads to
\begin{equation}
\delta \leq \|B(t)\|_{\mathrm{F}}\sqrt{1-\mathrm{IPR}(\ket{\Psi})}\,,
\end{equation}
and therefore (assuming $\delta\leq \|B(t)\|_{\mathrm{F}}$ so the right-hand side is nonnegative)
\begin{equation}
\mathrm{IPR}(\ket{\Psi}) \leq 1-\frac{\delta^2}{\|B(t)\|_{\mathrm{F}}^2}\,.
\end{equation}
This proves the first inequality in \eqref{eq:ipr_condition_CM}.

Moreover, on an $N$-dimensional donor space, we have
\begin{equation}
\|B(t)\|_{\mathrm{F}}\le \sqrt{N}\,\|B(t)\|_\infty =\sqrt{N}\,\mathcal{C}_{M_{\mathrm{trap}}(t)}(\mathrm{id})\,,
\end{equation}
hence $\|B(t)\|_{\mathrm{F}}^2\le N\,\mathcal{C}_{M_{\mathrm{trap}}(t)}(\mathrm{id})^2$, and thus
\begin{equation}
\mathrm{IPR}(\ket{\Psi}) \leq 1-\frac{\delta^2}{\|B(t)\|_{\mathrm{F}}^2}  \leq  1 - \frac{\delta^2}{N\,\mathcal{C}_{M_{\mathrm{trap}}(t)}(\mathrm{id})^2}\,.
\end{equation}
Since $\mathrm{PR}(\ket{\Psi})=1/\mathrm{IPR}(\ket{\Psi})$, the above IPR upper bound implies
\begin{equation}
\mathrm{PR}(\ket{\Psi}) \geq \frac{1}{1-\delta^2/\|B(t)\|_{\mathrm{F}}^2}\geq \frac{1}{1-\delta^2/(N\,\mathcal{C}_{M_{\mathrm{trap}}(t)}(\mathrm{id})^2)}\,.
\end{equation}
\end{proof}

\section{Proof of Theorem \ref{thm:lightcone} \label{app:thm-lightcone}}

\begin{proof}
Let $t\geq 0$ be fixed and define $\tau_t:=\Lambda_t^\dagger$, and the Hermitian operator $B(t):=(\mathrm{id}-\mathcal{G}^\dagger)\tau_t(M)$, which is Hermitian because $M$ is Hermitian and $\tau_t$ and $\mathcal{G}^\dagger$ preserve Hermiticity. 
Since $P_S$ is an orthogonal projector, it follows that
\begin{equation}\label{eq:step1}
\mathcal C_M^{(S)}(\Lambda_t) =  \|P_S B(t) P_S\|_\infty\,.
\end{equation}
Moreover, it holds that
\begin{equation}
\|(\mathrm{id} - \mathcal{G}^\dagger)(\tau_t(M))\|_{\infty}\leq 2 \|\tau_t(M)\|_\infty
\end{equation}
such that
\begin{equation}\label{eq:step2}
\|P_S B(t) P_S\|_\infty \leq \|B(t)\|_\infty \leq 2\|\tau_t(M)\|_\infty\,.
\end{equation}
To exploit locality, we now use the quasi-locality assumption \eqref{eq:quasilocality_assumption}
for the observable $M$ supported on $X$.
Let $d:=d(S,X)\ge 1$ and choose $r:=d-1$. Then $S\cap X[r]=\emptyset$, hence for any operator
$O$ supported on $X[r]$ we have $P_S O P_S=0$.
Let $M^{(r)}(t)$ be the operator supported on $X[r]$ whose existence is guaranteed by
\eqref{eq:quasilocality_assumption}. Since $P_S M^{(r)}(t) P_S=0$,
\begin{equation}
P_S\tau_t(M)P_S = P_S\big(\tau_t(M)-M^{(r)}(t)\big)P_S.
\end{equation}
Taking norms and using \eqref{eq:quasilocality_assumption} gives
\begin{align}\label{eq:step3}
\|P_S\tau_t(M)P_S\|_\infty &\leq \|\tau_t(M)-M^{(r)}(t)\|_\infty\nonumber\\
&\leq C_{\mathrm{QL}}\|M\|_\infty \exp \big(v_{\mathrm{QL}}t-\mu(d-1)\big)\,.
\end{align}
Combining \eqref{eq:step1} and \eqref{eq:step3} finally yields
\begin{align}
\mathcal C_M^{(S)}(\Lambda_t) &\leq 2 C_{\mathrm{QL}}\,\|M\|_\infty\,\exp \big(v_{\mathrm{QL}}t-\mu(d-1)\big)\nonumber\\
&= \bigl(2C_{\mathrm{QL}}e^{\mu}\bigr)\,\|M\|_\infty\,\exp \big(v_{\mathrm{QL}}t-\mu d\big)\,,
\end{align}
which is Eq.~\eqref{eq:lightcone_bound} with $C_{\mathrm{LC}}:=2C_{\mathrm{QL}}e^\mu$.
\end{proof}

\bibliography{Refs2}

\end{document}